\documentclass[a4paper,11pt]{article}
\usepackage{cite}

\usepackage{algorithm}
\usepackage[noend]{my_algorithmic}

\usepackage{graphicx}
\usepackage{amsfonts}
\usepackage{amsmath}
\usepackage{amsthm}
\usepackage{xspace}
\usepackage{subfigure}
\usepackage{hyperref}

\newtheorem{theorem}{Theorem}
\newtheorem{observation}{Observation}
\newtheorem{corollary}[theorem]{Corollary}
\newtheorem{proposition}[theorem]{Proposition}
\newtheorem{lemma}{Lemma}

\newtheorem{definition}{Definition}
\newtheorem{assumption}{Assumption}

\newcommand{\First}[1]{{\it First}\ensuremath{(#1)}\xspace}
\newcommand{\Last}[1]{{\it Last}\ensuremath{(#1)}\xspace}
\newcommand{\dist}{\ensuremath{{\rm dist}}\xspace}

\newcommand{\IC}{\ensuremath{\mathcal I}\xspace}

\newcommand{\Look}{{\it Look}\xspace}
\newcommand{\Compute}{{\it Compute}\xspace}
\newcommand{\Move}{{\it Move}\xspace}

\newcommand{\MoveTo}[1]{{\tt Move(#1)}\xspace}

\newcommand{\Pos}[1]{{{\tt Pos[}#1{\tt]}}\xspace}
\newcommand{\DP}[1]{{{\tt DP}(#1)}\xspace}
\newcommand{\setRobots}{\ensuremath{{\mathcal R}}\xspace}

\newcommand{\MS}[1]{\ensuremath{{\mathcal M}{\mathcal S}(#1)}\xspace}

\newcommand{\LIM}{{{\tt LIM}}\xspace}
\newcommand{\AW}[1]{{{\tt AW}(#1)}\xspace}

\newcommand{\Qi}{\ensuremath{Q_1}\xspace}
\newcommand{\Qii}{\ensuremath{Q_2}\xspace}
\newcommand{\NW}{\ensuremath{{\mathcal N}{\mathcal W}}\xspace}

\newcommand{\SE}{\ensuremath{{\mathcal S}{\mathcal E}}\xspace}

\providecommand{\norm}[1]{\lVert#1\rVert}

\newcommand{\NearG}{{\sc Near-Gathering}\xspace}

\newcommand{\ASYNC}{{\sc Async}\xspace}
\newcommand{\SSYNC}{{\sc Ssync}\xspace}
\newcommand{\FSYNC}{{\sc Fsync}\xspace}

\newenvironment{state}[2]{
\fbox{
\begin{minipage}{\columnwidth}
{\bf State} #1
\ \\
{\small #2}
\end{minipage}
}
}{\vspace{1em}}

\title{Getting Close Without Touching: Near-Gathering for \\Autonomous Mobile Robots\thanks{This work has been partially supported by MIUR of Italy under project ARS TechnoMedia.}}

\author{
Linda Pagli\footnotemark[1]\
\and
Giuseppe Prencipe\footnotemark[1]\
\and
Giovanni Viglietta\footnotemark[2]\
}
\date{}
\begin{document}

\maketitle

\renewcommand{\thefootnote}{\fnsymbol{footnote}}
\footnotetext[1]{Dipartimento di Informatica,
Universit\`a di Pisa, {\tt \{linda.pagli,giuseppe.prencipe\}@unipi.it}}
\footnotetext[2]{School of Electrical Engineering and Computer Science,
University of Ottawa, Canada, {\tt viglietta@gmail.com}}
\renewcommand{\thefootnote}{\arabic{footnote}}

\maketitle

\begin{abstract} 
In this paper we study the \NearG problem for a finite set of dimensionless, deterministic, asynchronous, anonymous, oblivious and autonomous mobile robots with limited visibility moving in the Euclidean plane in Look-Compute-Move (LCM) cycles. In this problem, the robots have to get close enough to each other, so that every robot can see all the others, without touching (i.e., colliding with) any other robot. The importance of solving the \NearG problem is that it makes it possible to overcome the restriction of having robots with limited visibility. Hence it allows to exploit all the studies (the majority, actually) done on this topic in the unlimited visibility setting. Indeed, after the robots get close enough to each other, they are able to see all the robots in the system, a scenario that is similar to the one where the robots have unlimited visibility. 

We present the first (deterministic) algorithm for the \NearG problem, to the best of our knowledge, which allows a set of autonomous mobile robots to nearly gather within finite time without ever colliding. Our algorithm assumes some reasonable conditions on the input configuration (the \NearG problem is easily seen to be unsolvable in general). Further, all the robots are assumed to have a compass (hence they agree on the ``North'' direction), but they do not necessarily have the same handedness (hence they may disagree on the clockwise direction).

We also show how the robots can detect termination, i.e., detect when the \NearG problem has been solved. This is crucial when the robots have to perform a generic task after having nearly gathered. We show that termination detection can be obtained even if the total number of robots is unknown to the robots themselves (i.e., it is not a parameter of the algorithm), and robots have no way to explicitly communicate.
\end{abstract}

\section{Introduction}

Consider a distributed system whose entities are a finite set of {dimensionless} {\em robots} or {\em agents} that can freely move on {the Euclidean} plane, operating in {\em Look-Compute-Move} (LCM) cycles. During each cycle, a robot takes a snapshot of the positions of the other robots (\Look);
executes a deterministic protocol, the same for all robots, using the snapshot as an input (\Compute); and
moves towards the computed destination (\Move). After each cycle, a robot may {stay idle} for some time.
With respect to the LCM cycles, the most common models used in these studies are  the {\em  fully synchronous} (\FSYNC),  
the {\em semi-synchronous} (\SSYNC),  and  the {\em asynchronous} (\ASYNC).   
In the {\em asynchronous} (\ASYNC) model, each robot acts independently from the others and the duration of each cycle is finite but unpredictable; thus, there is no common notion of time, and robots can compute and move based on ``obsolete'' observations. 
In contrast, in the {\em fully synchronous} (\FSYNC) model, there is a common notion of time, and robots execute their cycles synchronously. In {this model}, time is assumed to be discrete, and at each time instant {\em all} robots are activated, obtain the same snapshot, compute and move towards the computed destination; thus, no computation or move can be made based on obsolete observations.  
The last model, the {\em semi-synchronous} (\SSYNC), is like \FSYNC where, however, not all robots are necessarily activated at each time instant. 

In the last few years, the study of the computational capabilities of such a system has gained much attention, and the main goal of the research efforts has been to understand the relationships between the capabilities of the robots and their power to solve common tasks. The main capabilities of the robots that, to our knowledge, have been studied so far in this distributed setting are {\em visibility}, {\em memory}, {\em orientation}, and {\em direct communication}. With respect to visibility, the robots can either have {\em unlimited visibility}, if they sense the positions of {\em all} other robots, or have {\em limited visibility}, if they sense just a portion of the plane, up to a given distance $V$~\cite{andoi,FloPSW05}. With respect to memory, the robots can either be {\em oblivious}, if they have access only to the information sensed or computed during the current cycle (e.g.,~\cite{SouDY09}), or {\em non-oblivious}, if they have the capability to store the information sensed or computed since the beginning of the computation (e.g.,~\cite{cie04,SuzY99,yamashita2010}). With respect to orientation, the two extreme settings studied are the one where the robots have {\em total agreement}, and agree on the orientation and direction of their local coordinate systems (i.e., they agree on a {\em compass}), e.g.,~\cite{FloPSW08}, and the one where the robots have {\em no agreement} on their local coordinate axes, e.g.,~\cite{SuzY99,yamashita2010}. In the literature, there are studies that tackle also the scenarios in between; for instance, when the robots agree on the direction of only one axis, or there is agreement just on the orientation of the coordinate system (i.e., right-handed or left-handed), e.g.,~\cite{petitLeaderPF}. With respect to direct communication, {some recent studies introduced} the use of external signals or lights to enhance the capabilities of mobile robots. These were first suggested in~\cite{pelegInvited}, and were also referenced in~\cite{efrima07}, which provided the earliest indication that incorporating some simple means of signaling in the robot model might positively affect the power of the team. Recently, a study that tackles this particular capability more systematically has been presented in~\cite{floccICDCS2012}. 

In this paper, we solve the \NearG problem: the robots are required to get close enough to each other, without ever colliding during their movements. Here, the team of robots under study executes the cycles according to the \ASYNC model, the robots are oblivious and have limited visibility. The importance of {solving the \NearG problem} is that {it allows} to overcome the limitations of having robots with limited visibility, and it {makes it} possible to exploit all the studies (the majority, actually) done in the unlimited visibility setting, such as, for instance, the {\em Arbitrary Pattern Formation Problem}~\cite{petitLeaderPF,FloPSW08,SuzY99,yamashita2010}, or the {\em Uniform Circle Formation} (e.g., \cite{DefS08,DieLP08}). {Indeed, if all the robots get close enough, they eventually become able to see one another, reaching a configuration in which they may be assumed to have unlimited visibility (recall that the robots are dimensionless).} Since most of the {studies related} to the unlimited visibility case assume a starting configuration where no two robots {coincide} (i.e., they do not share the same location in the plane), it is of crucial importance to ensure that no collision occurs during the process.

A problem {that is similar} to \NearG is the {\em gathering} problem, {in which} the robots have to meet, within finite time, in a point of the plane not agreed upon in advance. {Note that the gathering problem requires all robots to actually become coincident, while in \NearG they have to \emph{approach} a point, but they are not allowed to collide with each other. Another related problem is the \emph{convergence} problem, in which the robots have only to approach a point in the plane and converge to it in the limit, but they do not necessarily have to reach it in finite time, and they may collide with each other in the process. Hence, the convergence problem is easier than both gathering and \NearG. For a discussion on previous solutions to the problems of gathering and convergence, and how they fail to solve \NearG, refer to Section~\ref{sec:previous}.}

A preliminary solution to the \NearG problem has been presented {by the authors} in~\cite{NearGSIROCCO}; however, that solution worked with distances induced by the infinity norm.\footnote{{The infinity norm of a vector $(x,y)\in \mathbb R^2$ is defined as $\norm{(x,y)}_\infty=\max\{|x|,|y|\}$.}} In this paper we drop that assumption, presenting a more general solution that works with the usual Euclidean distance. We emphasize that the technique used in this paper can be easily adapted to solve the \NearG problem under any $p$-norm distance with $p \geq 1$, including the infinity norm distance used in~\cite{NearGSIROCCO}. We also note that, in contrast with \cite{NearGSIROCCO} and other works on limited visibility, such as \cite{FloPSW05}, we only assume that the robots have agreement on one axis (as opposed to both axes). In order to detect termination, the algorithm in~\cite{NearGSIROCCO} requires either the knowledge of the number of robots in the system, or the ability of the robots to communicate through visible lights that can be turned on or off. In the present paper we are able to drop both requirements, and still detect termination.

It is worth mentioning that in~\cite{NearGSIROCCO} a tacit assumption is made on the starting positions of the robots. Namely, we consider the graph on the robot set, with an edge connecting two robots if their initial distance is at most $D$, where $D$ is a known constant that is smaller than the visibility radius of the robots (but may be arbitrarily close to it). The assumption is that such a graph is connected. Here we make this assumption explicit, and we give a more rigorous proof of our algorithm's correctness. Finally, we remark that, since the algorithm presented here is for the \ASYNC model, it solves the problem \emph{a fortiori} also in the \SSYNC and \FSYNC models.

The organization of the paper is as follows: in Section~\ref{sec:model} the formal definition of the robot model is presented; in Section~\ref{sec:solution} the collision-free algorithm that solves the \NearG problem is presented, {after discussing why previous solutions to related problems fail to solve it}; in Section~\ref{sec:correctness} the correctness of {our algorithm is proven}. Finally, Section~\ref{sec:conclusions} concludes the paper, {suggesting some directions for future research}.

\begin{figure}[t]
\centering
\subfigure[]{\label{fig:axes}\includegraphics[scale=.85]{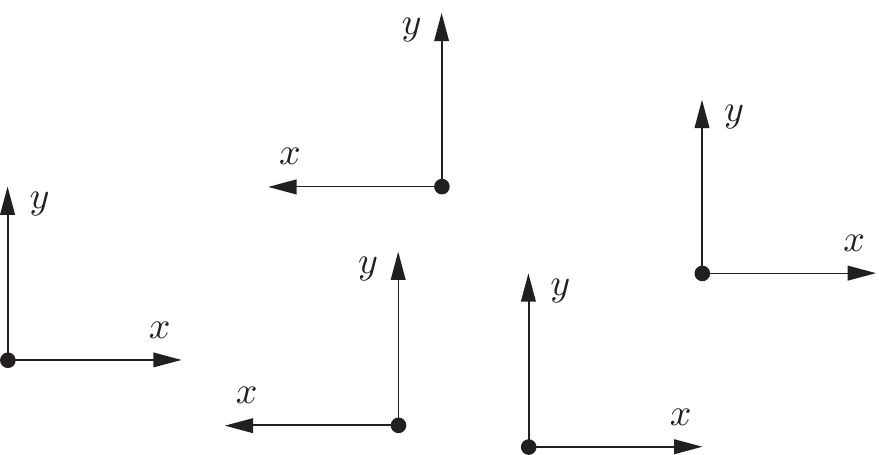}}\qquad\qquad\qquad
\subfigure[]{\label{fig:algoa}\includegraphics[scale=.7]{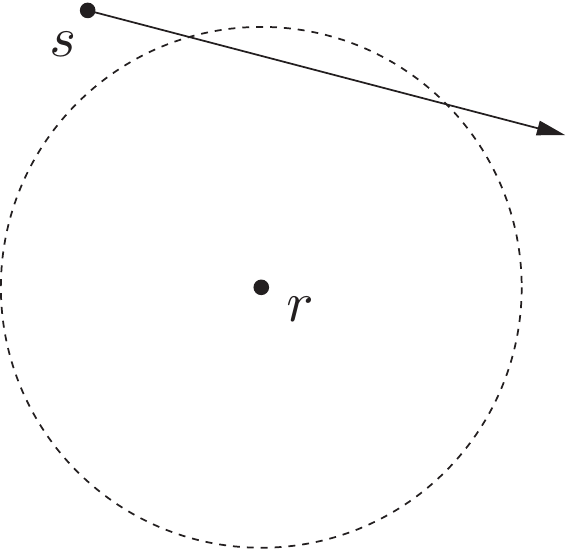}}
\caption{{(a) The robots in the swarm agree on the $y$-axis but not on the $x$-axis.} (b) In the limited visibility setting a robot can only see robots that are within its radius of visibility. As a consequence, when $s$ starts moving
(the left end of the arrow), $r$ and $s$ do not
see each other. While $s$ is moving, perhaps $r$ \Look{}s and sees $s$;
however, $s$ is still unaware of $r$. After
$s$ passes the area of visibility of $r$, it is still unaware of $r$.}
\end{figure}

\section{The Model\label{sec:model}} 

The system is composed of a team of {finitely many} mobile entities, called {\em robots}, each {representing} a computational unit provided with its own local memory and capable of performing local computations. The robots are {modeled as} points in the {Euclidean plane $\mathbb R^2$}. Let $r(t)$ denote
the ``absolute'' position of robot $r$ at time $t$ (i.e., with respect to an absolute {coordinate system}), {where $0\leq t\in \mathbb R$}; also, we will denote by $r(t).x$ and $r(t).y$ the abscissa and the ordinate value of $r(t)$, respectively. When no ambiguity arises, we shall omit the temporal indication; also, the {\em configuration} of the robots at time $t$ is the set of robots' positions at time $t$.

Each robot has its own local {orthogonal} coordinate system, {centered at its location,} and we assume that the local coordinate systems of the robots agree on the directions of the $x$- and $y$-axes. As discussed in Section~\ref{sec:conclusions}, the algorithms that we present in this paper also works in the more restricted model in which the robots agree on the direction of just one axis, {as illustrated in Figure~\ref{fig:axes}}. A robot is endowed with sensorial capabilities and it observes the world by activating its sensors, which return a snapshot of the positions of all other robots with respect to its local coordinate system. The visibility radius of the robots is limited: robots can sense only points in the plane within distance $V$. This setting, referred to in the literature as
{\em limited visibility}, is understandably more
difficult; for example, a robot with limited visibility might not even know the total
number of robots nor where they are located, if outside its visibility range. Also, when combined with the asynchronous behavior of the robots, it introduces a higher level of difficulty in the design of collision-free protocols. 
For instance, in the example depicted in Figure~\ref{fig:algoa}, robot $s$, in transit towards its destination, might be seen by $r$; however, $s$ is not aware of $r$'s existence
and, if it starts the next cycle before $r$ starts moving, $s$ will continue to
be unaware of $r$; hence, since $r$ does not see $s$ when $s$ starts its movement, it must take care of the possible arrival of $s$ when computing its destination. 

All robots are identical: they are indistinguishable from their appearance and they execute the same protocol. Robots are autonomous, without a central control. Robots are silent, in the sense that they have no means of direct communication (e.g., radio, infrared) of information to other robots. Robots are endowed with motorial capabilities, and can move freely in the plane. {As a robot moves, its coordinate system is translated accordingly, in such a way the the robot's location is always at the origin.}

{Each robot continually performs \Look-\Compute-\Move (LCM) cycles, each consisting of three different \emph{phases}:}

\begin{description}
\item[(i)] {\bf Look:} The robot observes the world by activating its sensor,
which returns a snapshot of
the positions of all robots within its radius of visibility with respect to its own
coordinate system
(since robots are {modeled} as points, their positions in the
plane are just the set of their coordinates).

\item[(ii)] {\bf Compute:} The robot executes its (deterministic) algorithm, using the snapshot as input.
The result of the computation is a destination point, {expressed in the robot's own coordinate system}. There is no time limit to perform such a computation, although the robot can only compute finite sequences of algebraic functions on the visible robots' coordinates (actually, the algorithm proposed in this paper uses only arithmetic operations and square roots).

\item[(iii)] {\bf Move:} The robot moves monotonically towards the computed destination along a straight line;
 if the destination is
the current location, the robot stays still (performs a {\em null movement}). No assumptions are made on the speed of the robot, as it may vary arbitrarily throughout the whole phase.
\end{description}

The robots do not have persistent memory, that is, memory whose content is preserved from one cycle to the next; they are said to be {\em oblivious}. The only available memory they have is used to store local variables needed to execute the algorithm, which are erased at each cycle. {All robots are initially idle, until they are activated by a scheduler and start executing the \Look phase of the first cycle.} The amount of time to complete a cycle is assumed to be finite, but unpredictably variable from cycle to cycle and from robot to robot {(i.e., the scheduler model is \ASYNC)}, but the \Look phase is assumed to be instantaneous. {As a consequence, a robot may even stay still for a long time after it has reached its current destination point, before performing the \Look phase of the next cycle, or it can stop for a while in the middle of a move and then proceed, etc. All these actions are controlled by the scheduler, which is an entity independent of the robots and their protocol, and may be seen as an ``adversary'' whose purpose is to prevent the robots from accomplishing their task.}

{The scheduler may also end the \Move phase of a robot before it has reached its destination, forcing it to start a new cycle with a new input and a new destination: this feature is intended to model, for instance, a limit to a robot's motion energy. However, there exists a constant $\delta > 0$ such that, if the destination point computed by a robot has distance smaller than $\delta$ from the robot's current location, the robot is guaranteed to reach it; otherwise, it will move towards it by at least $\delta$. Note that, without this assumption, the scheduler could make it impossible for a robot to ever reach its destination, even if the robot keeps computing the same destination point. For instance, the scheduler may force the robot to move by smaller and smaller amounts at every cycle, converging to a point that is not the robot's intended destination. Instead, if the robot cannot be interrupted by the scheduler before it has moved by at least $\delta$, and it keeps computing the same destination point, it is guaranteed to reach it in finitely many cycles. The value of $\delta$ is not known to the robots, hence it cannot be used in their computations}.

We will denote by ${\mathbb L(t)}$, ${\mathbb C(t)}$, ${\mathbb M(t)}$ the sets of robots that are, respectively, active in a \Look phase, in a \emph{Compute} phase, and in a \emph{Move} phase at time $t$.

{We stress that robots are modeled as just points in the plane, and as such they do not have an associated vector indicating their ``heading'' or ``forward direction''. Likewise, a robot's coordinate system never rotates, but only translates following the robot's movements. Moreover, all robots have the same \emph{visibility radius} $V$, which is known to them and can be used in their computations. $V$ also serves as a common unit distance for the robots.}

\subsection{Notation and Assumptions\label{sec:notation}}

We will denote by  $\setRobots=\{r_1, \cdots, r_n\}$ the set of robots in the system. The purpose of this paper is to study the \NearG problem:

\begin{definition}[\NearG]
The \emph{\NearG problem} requires all robots to terminate their execution in a configuration such that there exists a disk of radius $\varepsilon$ containing all the robots, where $\varepsilon$ is a fixed constant, and no two robots occupy the same location.
\end{definition}

All the robots are required to execute the same protocol during their \Compute phase. The input to such a protocol is the snapshot of the robots' locations obtained by the executing robot during its previous \Look phase, along with the visibility radius $V$ (which is the same for all robots), and of course the value of $\varepsilon$.

The protocol executed by the robots must be independent of the initial configuration of the robots, and must make the robots solve the \NearG problem from any initial configuration. However, in the limited visibility model, this requirement is known to be too strong, and some additional assumptions must be made on the \emph{initial distance graph} in order to make the problem solvable.

\begin{definition}[Initial Distance Graph~\cite{FloPSW05}]
The {\em initial distance graph} $I = (\setRobots, E)$ of the robots is the graph such that, for any two distinct robots $r$ and $s$, $\{r,s\} \in E$ if and only if $r$ and $s$ are initially at distance not greater than the visibility radius $V$, i.e., $\dist(r(0),\,s(0)) \leq V$.
\end{definition}

By ``$\dist$'' we denote the usual Euclidean distance. In~\cite{FloPSW05} it is proven that, if the initial distance graph $I$ is not connected, then the gathering problem may be unsolvable; the same result clearly holds also for the \NearG problem:

\begin{observation}
If the initial distance graph $I$ is not connected, the \NearG problem may be unsolvable.\qed
\end{observation}

However, our solution to the \NearG problem requires a slightly more restrictive initial condition. Let $\sigma$ be an arbitrary small and positive constant, and let $D = V-\sigma$.
\begin{definition}[Initial Strong Distance Graph]
The \emph{initial strong distance graph} $J = (\setRobots, E)$ of the robots is the graph such that, for any two distinct robots $r$ and $s$, $\{r,s\} \in E$ if and only if $r$ and $s$ are initially at distance not greater than $D$, i.e., $\dist(r(0),\,s(0)) \leq D$.
\end{definition}

In the following, we will assume that:
 
\begin{assumption}
\label{a:assumption1}
The initial strong distance graph $J$ is connected.
\end{assumption}

We remark that $D$ (or at least lower bound on $D$) must be known to the robots. This is not much of a benefit to the robots in terms of raw computational power, since $V$ is already known to all the robots and can already be used in their computations as a common unit distance. Besides, by choosing $\sigma$ to be small enough, the set of initial configurations ruled out by Assumption~\ref{a:assumption1} becomes negligible.

The reasons why we need such a slightly more restrictive assumption are technical, and will become apparent in Section~\ref{sec:correctness}, when the correcntess of our algorithm will be proven. We stress that Assumption~\ref{a:assumption1} only refers to the \emph{initial} strong distance graph, while it does not require such a graph to be connected at all times. However, as we will prove in Section~\ref{sec:correctness}, our algorithm will indeed preserve the connectedness of a closely related distance graph throughout the execution.

Note that the definition of \NearG does not require the robots to avoid collisions during the execution, but it only requires them to occupy distinct locations when they all have terminated their execution. However, for \NearG to be solvable, the robots must necessarily occupy distinct locations in the initial configuration, otherwise the scheduler could always activate coinciding robots simultaneously, and never allow them to occupy distinct locations. The algorithm we will describe in this paper is in fact collision-free, that is, it always prevents robots from colliding, provided that they start from distinct locations. As a by-product, our algorithm works regardless of the ability of the robots to detect the presence of more than one robot in the same location (called \emph{multiplicity detection} in the literature~\cite{cie04,FloPSW05}).

Another necessary assumption is that no robot is moving at time $t=0$. If the robots are already moving when the execution starts, and two robots have the same destination point, nothing can prevent them from colliding. Moreover, after they have collided, the scheduler can force them to remain coincident forever, by activating them synchronously. If this happens, the \NearG cannot be solved.

Summarizing, we will make Assumption~\ref{a:assumption1} on the initial configuration of the robots, and we will also assume that initially no robot is moving, and no two robots occupy the same location. The protocol executed by the robots in the \Compute phase takes this input:
\begin{itemize}
\item an array of points expressed in the local coordinate system of the executing robot, denoting the locations of the visible robots observed during the previous \Look phase;
\item the visibility radius $V$ (the same for all robots);
\item the value of $D$ (the same for all robots);
\item the value of $\varepsilon$ (required for termination).
\end{itemize}

Observe that the value of $\delta$ is \emph{not} part of the input, and therefore the robots do not have a lower bound on the minimum distance that they are guaranteed to cover in a single \Move phase.

\section{The \NearG Problem and Its Solution\label{sec:solution}}

{In Section~\ref{sec:previous} we discuss some previous solutions to the gathering and convergence problems, explaining why they cannot be easily adapted to solve \NearG. Then, in Section~\ref{sec:algo} we give our solution to the \NearG problem.

\subsection{Previous Solutions to Related Problems}\label{sec:previous}

\paragraph{Gathering.}
Of course, since the gathering problem requires all robots to collide, no solution to this problem is a valid solution to \NearG. However, we may wonder if a simple modification of an existing gathering algorithm may solve \NearG.

The gathering problem has been studied in the literature in all models but, to the best of our knowledge, the most pertinent paper is~\cite{FloPSW05}, which considers robots with limited visibility in the \ASYNC setting. The algorithm in~\cite{FloPSW05} assumes all robots to agree on the direction of both axes, and ideally it makes the leftmost and topmost robots move first, rightwards and downwards, until all the robots gather. According to the protocol, a robot $r$ will occasionally compute a destination point that coincides with another visible robot $s$'s location. To avoid this type of move, we may make $r$ move toward $s$ without reaching it. If we consider an initial configuration in which all robots lie on the same vertical line, the only robot that is allowed to move according to the algorithm in~\cite{FloPSW05} is the topmost robot $r$. Moreover, if $r$ moves downward without ever reaching the next robot, then no robot other than $r$ will ever be able to move. Therefore, we ought to let robots other than $r$ move, as well. Unfortunately, the proof of correctness of the algorithm, given in~\cite{FloPSW05}, strongly depends on the fact that the robots in the swarm move in a strictly ordered fashion. If we let any robot move, then we have to make sure that the visibility graph remains connected throughout the execution, and that the robots still converge to a single point. Clearly, even if a suitable adaptation of this idea can be effectively applied to solve \NearG, the modified protocol would require a radically new analysis and proof of correctness.

\paragraph{Convergence.}
Several solutions to the convergence problem have been proposed, as well. If we manage to obtain a solution that also avoids collisions, we can successfully apply it to \NearG.

Perhaps the most natural strategy, at least in the unlimited visibility model, is to make all robots move to their center of gravity. This simple protocol has been analyzed in~\cite{CohP05}, and it has been proven correct even in the \ASYNC model. In the limited visibility setting, the only relevant work, to the best of out knowledge, is~\cite{andoi}, which gives a convergence algorithm that assumes the \SSYNC scheduler. However, in the special case in which the robots' locations are the vertices of a regular polygon and they are all mutually visible, both the center-of-gravity algorithm and the algorithm in~\cite{andoi} behave in the same way, and make any active robot move to the center of the polygon. Hence, if two robots are activated simultaneously from this configuration, they collide and fail to solve \NearG.

Therefore, we may modify the protocol and make each robot approach the center of gravity by, say, moving half-way towards it. We show that this protocol may still cause collisions in the \ASYNC model, even in the very simple case in which the system consists of only two robots. Let $r$ and $s$ be two mutually visible robots, such that $r(0)=(0,0)$ and $s(0)=(3124,0)$. Let the scheduler activate $r$, which observes that the center of gravity is point $(1562,0)$, and therefore computes the destination point $(781,0)$ (i.e., the point half-way toward the center of gravity). Now the scheduler lets $r$ start moving and, as soon as it reaches point $(52,0)$, it temporarily delays the remaining part of the move and makes $s$ quickly perform five complete cycles. As $r$ is always seen in $(52,0)$, $s$ moves first to $(2356,0)$, then to $(1780,0)$, $(1348,0)$, $(1024,0)$, and finally to $(781,0)$. Now the scheduler lets $r$ finish its original move, and this causes a collision with $s$ in $(781,0)$. Observe that, even if the protocol does not make the robots move half-way toward the center of gravity, but to some other fraction of the distance, similar examples can be constructed in which the robots collide.

\paragraph{Further literature.}
Several other papers considered the gathering or the convergence problems in various models, but these results are either not relevant to \NearG in our model, or they can be reduced to solutions already discussed above, and therefore discarded.

In~\cite{SouDY09}, the gathering problem is studied for robots with limited visibility, the \SSYNC scheduler, and \emph{temporarily unreliable compasses}. In the special case in which the robots are close enough and their compasses are reliable, the proposed algorithm becomes equivalent to that of~\cite{FloPSW05}, which has already been analyzed and discussed.

The gathering problem is studied in~\cite{ganguli09} in the context of non-convex environments and limited visibility, but with the \FSYNC scheduler. However, if the robots are close enough and they all see each other, the algorithm makes them all move to the center of the smallest enclosing circle. Hence, in the special case in which they form a small-enough regular polygon, they move to the center of gravity, and therefore the algorithm becomes equivalent to those of~\cite{andoi,CohP05}, which have already been discussed.

The convergence problem with limited visibility has been studied also for robots whose level of asynchronicity lies strictly between \SSYNC and \ASYNC. In~\cite{LinMA07b}, it is assumed that the time spent in a \Look or \Move phase is bounded, and the algorithm is a slight modification of that of~\cite{andoi}. In particular, it suffers from the issues that have already been discussed for~\cite{andoi}.

On the other hand, in~\cite{katreniak2011} the scheduler is \emph{1-bounded} \ASYNC, which means, roughly, that no robot can perform more than one \Look phase between two consecutive \Look phases of another robot. As it turns out, if the number of robots is even and they are vertices of a small-enough regular polygon, the algorithm makes them move to the center of gravity. Once again, this type of move has already been analyzed and discarded.

In~\cite{IIKO13}, the gathering problem is considered for the \SSYNC scheduler and the unlimited visibility setting. Here the focus is on the expected termination time of a randomized algorithm where the robots have some sort of \emph{multiplicity detection}, i.e., the ability to detect the presence of more than one robot in the same location. Unfortunately, both algorithms presented in this paper make all robots move to the center of the smallest enclosing circle, except in some special cases. When applied to the \NearG problem, this approach suffers from the same issues of the center-of-gravity approach.

In~\cite{pelc09}, the gathering problem in \ASYNC is studied for \emph{fat} robots, i.e., robots that are modeled as solid discs rather than dimensionless points. Unfortunately, the problem is solved only for a swarm of at most four robots, and the technique involves a case analysis that does not generalize to bigger swarms. Therefore this solution is irrelevant to our problem.

The above result has been generalized in~\cite{AGM13}, which solves the gathering problem for any number of fat robots. The robots considered have an unlimited visibility radius, and therefore the limitations posed by a bounded visibility radius are not addressed in the paper. Additionally, letting fat robots collide is not an issue, but instead it is a necessary event that is sought by the algorithm. For these reasons, the approach of this paper can hardly be adapted to our problem.

Another work that considers the gathering problem for fat robots is~\cite{CLDFH11}, which works in the unlimited visibility setting and the \FSYNC scheduler. Moreover, the gathering point is given as input to all the robots. Because of these differences with our model, it is impossible to extract a sound algorithm for \NearG from this work: indeed, the task of making such fat robots touch each other is simple, and the paper focuses on how to make robots slide around each other in order to occupy a small area. All these issues are meaningless in our model, and the real issues of our model become meaningless with fat robots.

\subsection{Solving the \NearG Problem}\label{sec:algo}

We conjecture that no solution to \NearG exists in the \ASYNC model in which the robots do not agree at least on one axis. Therefore, in the following we will assume to have agreement on both axes, and in Section~\ref{sec:conclusions} we will observe that our solution works even in the case of agreement on just one axis.}

The general high-level idea of the algorithm is to make the robots move upward and to the right, until they aggregate around the top-right corner of the smallest box that contains all of them. A robot's destination point is carefully computed, taking into account several factors. To avoid collisions, robots try to move in order, never ``passing'' each other, and never getting in each other's way. This is not a trivial task, because the visibility of the robots is limited, and they cannot predict the moves of the robots they cannot see. On the other hand, robots try to preserve mutual visibility by not moving too far from other visible robots, avoiding to leave them behind. This is supposed to prevent the robots from separating into different groups, which may be unaware of each other and aggregate around different points. As it turns out, the robots are unable to always preserve mutual visibility, but they can indeed preserve ``mutual awareness'', which is a concept that will be introduced shortly. These different behaviors are blended together and balanced in such a way that the robots are not only guaranteed to avoid collisions and remain mutually aware, but also to effectively aggregate around some point, and never ``get stuck'' or converge to different limit points. This is obtained by always making robots move by the greatest possible amount, compatibly with the above restrictions.

\begin{figure}[!ht]
\begin{state}{\Look}{
\vspace{-1em}
\begin{itemize}
\item[] Take the snapshot of the positions of the visible robots, which returns, for each robot $r \in \mathcal R$ at distance at most $V$, $\Pos{r}$, the position in the plane of robot $r$, according to my coordinate system (i.e., my position is $(0,0)$).
\end{itemize}}
\end{state}

\begin{state}{\Compute (returns destination point $dp=(dp.x,dp.y)$}{
\vspace{-1em}
\begin{algorithmic}
\STATE $\rho=\min\left\{V/4,\,V-D\right\}$;
\STATE $\varepsilon'=\min\{\varepsilon, \, \rho/2\}$;
\STATE $\mathcal Z = $ Set of visible robots (including myself);
\IFi{$\forall r_1,r_2 \in \mathcal Z,\ \dist(\Pos{r_1},\Pos{r_2})\leq \varepsilon'$}
{\textbf{Terminate};}
\STATE $D_0=$ Closed disk with radius $V$ and center in $(0,0)$;
\STATE $D_1=$ Closed disk with radius $V-\rho/2$ and center in $(0,0)$;
\STATE $D_2=$ Closed disk with radius $V-\rho$ and center in $(0,0)$;
\STATE $p_1=$ Leftmost intersection between $D_1$ and the horizontal line through $(0,\,V-\rho)$;
\STATE $p_2=$ Bottommost intersection between $D_1$ and the vertical line through $(V-\rho,\,0)$;
\STATE $S=$ Full closed square circumscribed around $D_2$ with edges parallel to the $x$- and $y$-axes;
\STATE $R=D_1 \cap S$;
\STATE $\Qi=$ Set of points of $D_0$ with positive $y$-coordinate and non-positive $x$-coordinate;
\STATE $\Qii=$ Set of points of $D_0$ with positive $x$-coordinate and non-positive $y$-coordinate;
\STATE $H_1=$ Set of points of $(R\setminus D_2)\cap \Qi$ whose $x$-coordinate is lower than $p_1.x$;
\STATE $H_2=$ Set of points of $(R\setminus D_2)\cap \Qii$ whose $y$-coordinate is lower than $p_2.y$;
\STATE $\NW = \left\{ r\in\mathcal Z\ |\ \Pos{r}\in \Qi \right\}$;
\STATE $\SE = \left\{ r\in\mathcal Z\ |\ \Pos{r}\in \Qii \right\}$;
\STATE $\displaystyle dp.x=\min\left\{\min_{\, r\in \mathcal \SE}\left\{\Pos{r}.x\right\},\ \max_{\, r\in \mathcal Z}\left\{\Pos{r}.x\right\},\, \rho/2\right\}$;
\STATE $\displaystyle dp.y=\min\left\{\min_{\, r\in \mathcal \NW}\left\{\Pos{r}.y\right\},\ \max_{\, r\in \mathcal Z}\left\{\Pos{r}.y\right\},\, \rho/2\right\}$;
\FOR{\textbf{Each} $r\in\mathcal Z$}
	\IFiELSIF{$\Pos{r}\in H_1$}{$dp.x=0$;\\}
	{$\Pos{r}\in R$ \textbf{Then}}
		\STATE $s_1=$ Leftmost intersection between $R\setminus H_1$ and the horizontal line through $\Pos{r}$;
		\STATE $\displaystyle dp.x=\min\left\{dp.x,\, \Pos{r}.x-s_1.x\right\}$;
	\ENDIF
	\IFiELSIF{$\Pos{r}\in H_2$}{$dp.y=0$;\\}
	{$\Pos{r}\in R$ \textbf{Then}}
		\STATE $s_2=$ Bottommost intersection between $R\setminus H_2$ and the vertical line through $\Pos{r}$;
		\STATE $\displaystyle dp.y=\min\left\{dp.y,\, \Pos{r}.y-s_2.y\right\}$;
	\ENDIF
\ENDFOR
\IFi{$dp.x>dp.y$}{$dp=\left(dp.x/2,\, 0\right)$; \textbf{Else} $dp=\left(0,\, dp.y/2\right)$;}
\end{algorithmic}}
\end{state}

\begin{state}{\Move}{
\vspace{-1em}
\begin{algorithmic}
	\STATE \MoveTo{$dp$}.
\end{algorithmic}}
\end{state}

\caption{The \NearG protocol \label{algo:1}}
\end{figure}

\begin{figure}[t]
\centering
\includegraphics[scale=1]{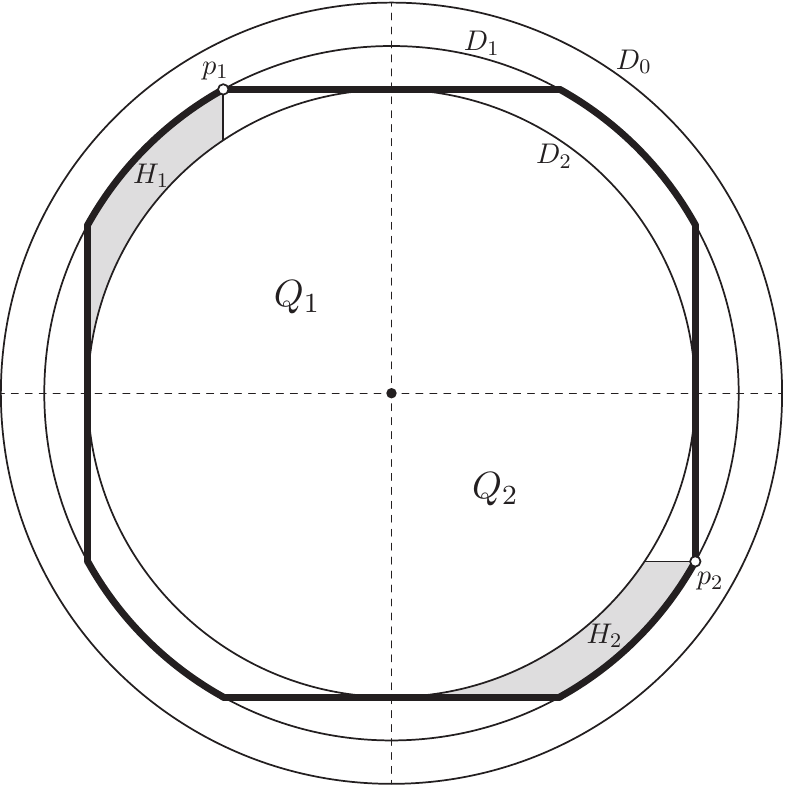}
\caption{Some of the elements computed by the \NearG protocol. The computing robot lies in the center, and the thick line represents the boundary of $R$.}
\label{fig:algob}
\end{figure}

The details of our \NearG protocol are reported in Figure~\ref{algo:1}. The protocol is executed by each robot during every \Compute phase. In the following, we denote by $r^*$ the robot that is currently executing the algorithm. The returned value is $dp$, which is the destination point for $r^*$. The algorithm computes separately the horizontal and the vertical components of the movement of $r^*$, i.e., $dp.x$ and $dp.y$. Note that the computation of the horizontal component $dp.x$ is symmetrical to the computation of the vertical component, hence any proposition that holds for the $x$ coordinate holds symmetrically for the $y$ coordinate.

Referring to the \NearG protocol and to Figure~\ref{fig:algob}, let $D_1$ and $D_2$ be the (closed) disks with radius $V-\rho/2$ and $V-\rho$, respectively, and center in the current position of $r^*$. Also, let $S$ be the closed square circumscribed around $D_2$ (with sides parallel to the $x$- and $y$-axes), and $R=D_1 \cap S$. Finally, let $H_1$ and $H_2$ be the {\em halt zones} of $r^*$, and \NW and \SE the sets of visible robots in \Qi and \Qii, respectively (note that \Qi contains its right border, but not the bottom one; similarly, \Qii contains its top border, but not the left one). 

Because no robot ever moves leftwards or downwards, we give the following definition:

\begin{definition}[Move Space]\label{def:movespace}
The {\em Move Space} of a robot $r$ at time $t$, denoted by \MS{r,t}, is the set $\left\{ (x',y')\in \mathbb R^2 \mid x'\geq r(t).x \wedge y'\geq r(t).y \right\}$.
\end{definition}

The destination point of $r^*$ is computed according to the rules below:

\begin{enumerate}
\item\label{rule1} $r^*$ only moves rightward or upward (not diagonally) at every move. It moves by the greatest possible amount, compatibly with the following restrictions (this is needed for the algorithm's convergence, see Section~\ref{sec:T}).
\item\label{rule2} $r^*$ never enters the \emph{move space} of a visible robot, unless it already is in its move space (this is required for collision avoidance, Section~\ref{sec:CA});
\item\label{rule3} $r^*$ never moves to the right of (resp.\ above) the rightmost (resp.\ topmost) robot it can see (needed for convergence, Section~\ref{sec:T});
\item\label{rule4} If $r^*$ sees a robot in the halt zone to its left (i.e., $H_1$ in Figure~\ref{fig:algob}), $r^*$ does not move rightward. Symmetrically, if $r^*$ sees a robot in the bottom halt zone (i.e., $H_2$ in Figure~\ref{fig:algob}), $r^*$ does not move upward. This is needed for the preservation of mutual awareness, see Section~\ref{sec:MA};
\item\label{rule5} If $r^*$ sees a robot $r$ in $R\setminus H_1$ (resp.\ $R\setminus H_2$), it moves so that $r$ stays inside $R\setminus H_1$ (resp.\ $R\setminus H_2$) (preservation of mutual awareness, Section~\ref{sec:MA}). Note that this does not guarantee \emph{a priori} that $r$ will actually stay inside $R\setminus H_1$ (resp.\ $R\setminus H_2$), since $r$ moves asynchronously and independently of $r^*$;
\item\label{rule6} The length of the so-computed movement is capped at $\rho/2$ (where $\rho=\min\left\{V/4,\,V-D\right\}$), and then halved (this is needed for both mutual awareness preservation and collision avoidance, see Sections~\ref{sec:MA} and~\ref{sec:CA}).
\end{enumerate}
To correctly detect termination, we make sure that $\varepsilon$ is not greater than $\rho/2$, by setting $\varepsilon'=\min\{\varepsilon, \, \rho/2\}$. This is necessary to prove Lemma~\ref{lem:lights_convergence}.

\section{Correctness\label{sec:correctness}}

In this section, we will prove that the protocol reported in Figure~\ref{algo:1} correctly solves the \NearG problem. The proof will be articulated in three parts: first, we will prove that a suitably-defined distance graph remains connected during the execution; second, we will prove that no collisions occur during the movements of the robots; finally, we show that all the robots converge to the same limit point, and correctly terminate their execution.

\subsection{Preliminary Definitions and Observations}

Before presenting the correctness proof, we will introduce a few preliminary definitions and observations. First, it is easy to observe the following:

\begin{observation}\label{obs:1}
No robot's $x$- or $y$-coordinate may ever decrease. No robot's $x$- and $y$-coordinates can both increase during the same move. Furthermore, a robot can move rightward (resp.\ upward) only if there is another robot strictly to the right of (resp.\ strictly above) its destination point.\qed
\end{observation}

\begin{observation}\label{obs:3}
During each cycle, a robot travels a distance not greater than $\rho/4\leq V/16$.\qed
\end{observation}

We may assume that, in the last line of the algorithm, if $dp.x$ and $dp.y$ are equal, then one of the two values $(dp.x/2,\, 0)$ and $(0,\, dp.y/2)$ is chosen arbitrarily as the destination point $dp$. With this assumption, the following holds.

\begin{observation}\label{obs:sym}
The algorithm is symmetric with respect to $x$- and $y$-coordinates.\qed
\end{observation}

\begin{definition}[First and Last]
Given a robot $r$,
let $\First{r,t} = \min\{t' > t| r  \in
{\mathbb L(t')} \}$ be the first time, after time $t$, at which
$r$ performs a \Look operation. Also, 
let $\Last{r,t} = \max\{t' \leq t| r  \in {\mathbb L(t')} \}$
be the last time, from the beginning up to time $t$, at which $r$ has performed a \Look operation; if $r$ has not performed a \Look yet, then $\Last{r,t}=0$.
\end{definition}

Now, we define the {\em destination point} of a robot at a time $t$ as follows:

\begin{definition}[Destination Point]
Given a  robots $r$, we define the {\em destination point} $\DP{r,t}$ of $r$ at time $t$ as follows: 
\begin{itemize}
\item If $r \in {\mathbb L(t)}$, then $\DP{r,t}$ is the point $dp$ as computed in the next \Compute phase after $t$ (in the current cycle).
\item If $r \in {\mathbb C(t)}$, then $\DP{r,t}$ is the point $dp$ as computed in the current \Compute phase.
\item If $r \in {\mathbb M(t)}$, then $\DP{r,t}$ is the point $dp$ as computed in the last \Compute phase before $t$ (in the current cycle).
\end{itemize}
\end{definition}

From the previous definition, we can state the following:

\begin{observation}
\label{obs:dpInvariant}
Let $r$ be a robot. During the time strictly between two consecutive \Look{s}, the destination point of $r$ does not change.
\end{observation}
\begin{proof}
Let $t$ be any time when $r$ executes a \Look; then, by definition, $\DP{r,t}$ is the point $dp$ as computed in the next \Compute phase after $t$ (in the current cycle). Also, the destination point does not change in the next \Compute and \Move phases of $r$. 
\end{proof}

The following proposition states a straightforward geometric fact (refer also to Figure~\ref{fig:algob}): among the segments contained in the annulus $D_1\setminus D_2$, with one endpoint on the boundary of $D_1$ and the other endpoint on the boundary of $D_2$, the shortest are those that are collinear with the center of $D_2$. This will be often used in conjunction with Observation~\ref{obs:3}, to show that robots cannot lose visibility to each other under certain conditions.

\begin{proposition}\label{prop1}
The length of a segment contained in the annulus $D_1\setminus D_2$, with one endpoint on the boundary of $D_1$ and the other endpoint on the boundary of $D_2$, is at least $\rho/2$.
\end{proposition}
\begin{proof}
Due to the rotational symmetry of the annulus, it is enough to prove the proposition for vertical segments only. The claim is equivalent to saying that, if $x\in[0,\, V-\rho]\subset\mathbb R$, then $$\sqrt{(V-\rho/2)^2-x^2}-\sqrt{(V-\rho)^2-x^2}\geq \rho/2.$$ Let $f(x)=\sqrt{(V-\rho/2)^2-x^2}-\sqrt{(V-\rho)^2-x^2}$. Then, $f(0)=\rho/2$, and $f(x)$ is monotonically increasing on $[0,\, V-\rho]$. Indeed, the derivative of $f(x)$ on $(0,\, V-\rho)$ is $$\frac{d}{dx}f(x)=x\left(\frac{1}{\sqrt{(V-\rho)^2-x^2}}-\frac{1}{\sqrt{(V-\rho/2)^2-x^2}}\right),$$
which is positive.
\end{proof}

Let \Qi be defined as in the \NearG protocol reported in Figure~\ref{algo:1}; in the following, we will denote by $\Qi(r,t)$ the set \Qi as robot $r$ would compute it if it were in a \Compute phase at time $t$ (this set is expressed in the global coordinate reference system). A similar notation will be used for the other sets and points computed in our protocol (e.g., $D_0$, $D_1$, $\Qii$, etc.).

\subsection{Preservation of Mutual Awareness\label{sec:MA}}

We define yet another notion of distance graph on the robots. This is useful, because in Corollary~\ref{cor:connectvisible} we will prove that this graph remains connected throughout the execution of our \NearG protocol.
\begin{definition}[Intermediate Distance Graph]
The {\em intermediate distance graph} at time $t\geq 0$ is the graph $G(t) = (\setRobots, E(t))$ such that, for any two distinct robots $r$ and $s$, $\{r,s\} \in E(t)$ if and only if $r$ and $s$ are at distance not greater than $V-\rho/2$ at time $t$, i.e., $\dist(r(t),\,s(t)) \leq V-\rho/2$, where $\rho=\min\{V/4,V-D\}$.
\end{definition}

Recall that, by assumption, the initial strong distance graph $J$ is connected. This implies that $G(0)$ is connected, because $V-\rho/2>D$, and hence $J\subseteq G(0)$. We will now prove that the connectedness of the intermediate distance graph is preserved during the entire execution of the algorithm. We will do so after introducing the notion of {\em mutual awareness}, in Definition~\ref{def:ma}.

First we define the auxiliary relation $\AW{p,q}$.

\begin{definition}
Given two points $p,q\in \mathbb R^2$, we denote by $\AW{p,q}$ the (symmetric) relation\footnote{By $\norm{a}_2=\sqrt{a.x^2+a.y^2}$ we denote the usual Euclidean norm; by $\norm{a}_\infty = \max\left\{|a.x|,\, |a.y|\right\}$ we denote the infinity norm.}
$$\norm{p-q}_2\leq V-\rho/2\ \wedge\ \norm{p-q}_\infty\leq V-\rho.$$
\end{definition}

A simple fact to observe is the following (recall that $r(t)$ denotes the position of robot $r$ at time $t$).
\begin{observation}
\label{obs:AWequiv}
For any two robots $r$ and $s$, $\AW{r(t),\, s(t)}$ is equivalent to $s(t)\in R(r,t)$, which is equivalent to $r(t)\in R(s,t)$.\qed
\end{observation}

Recall that, according to the algorithm, if a robot $r$ sees a robot $s$ in $R$, it will make its next move in such a way that $s$, as it was observed, does not exit $R$ (see how $s_1$ and $s_2$ are computed in the algorithm). This is stated in the next observation.
\begin{observation}
\label{obs:AWR}
If $r$ and $s$ are two robots, $r\in\mathbb L(t)$ and $\AW{r(t),\, s(t)}$, then $\AW{\DP{r,t},\, s(t)}$.\qed
\end{observation}

Before introducing the next lemmas, let us recall that $D_1$ and $D_2$ are the closed disks with radius $V-\rho/2$ and $V-\rho$, respectively, and center in $(0,0)$; $S$ is the full closed square circumscribed around $D_2$ with sides parallel to the $x$- and $y$-axes; and that $R=D_1 \cap S$ (refer to the \NearG protocol, and to Figure~\ref{fig:algob}).

The next two lemmas are technical, and will be used in the proof of Lemma~\ref{lem:aware}.

\begin{lemma}\label{lem:aw1}
Let two robots $r$ and $s$ be given, with $r\in\mathbb L(t)$. If $\AW{r(t),\, s(t)}$ and $\AW{r(t),\, \DP{s,t}}$, then $\AW{\DP{r,t},\, s(t)}$ and $\AW{\DP{r,t},\, \DP{s,t}}$.
\end{lemma}
\begin{proof}
From Observation~\ref{obs:AWR} it immediately follows that $\AW{\DP{r,t},\, s(t)}$. Next we prove that $\AW{\DP{r,t},\, \DP{s,t}}$.

Without loss of generality we may assume that $s$ is not moving horizontally at time $t$, that is, $s(t).x=\DP{s,t}.x$ and $0\leq \DP{s,t}.y-s(t).y \leq\rho/4$ (cf.~Observations~\ref{obs:1}--\ref{obs:sym}). Let $\Delta=\DP{r,t}-r(t)$; first observe that $\AW{\DP{r,t},\, \DP{s,t}}$ is equivalent to $\AW{r(t),\, \DP{s,t}-\Delta}$. Hence we have to prove that the point $\DP{s,t}-\Delta$ lies in $R(r,t)=R$, provided that $s(t)$ and $\DP{s,t}$ do.

If $\Delta$ is the null vector, there is nothing to prove. So, let us assume first that $\Delta.x=0$ and $\Delta.y>0$. Referring to Figure~\ref{fig:MAa}, and by the convexity of $R$, it is sufficient to prove that $s(t)-\Delta$ lies in $R$, which is equivalent to $\AW{\DP{r,t},\, s(t)}$, which has already been proven.

\begin{figure}[ht]
\centering
\subfigure[]{\label{fig:MAa}\includegraphics[scale=1]{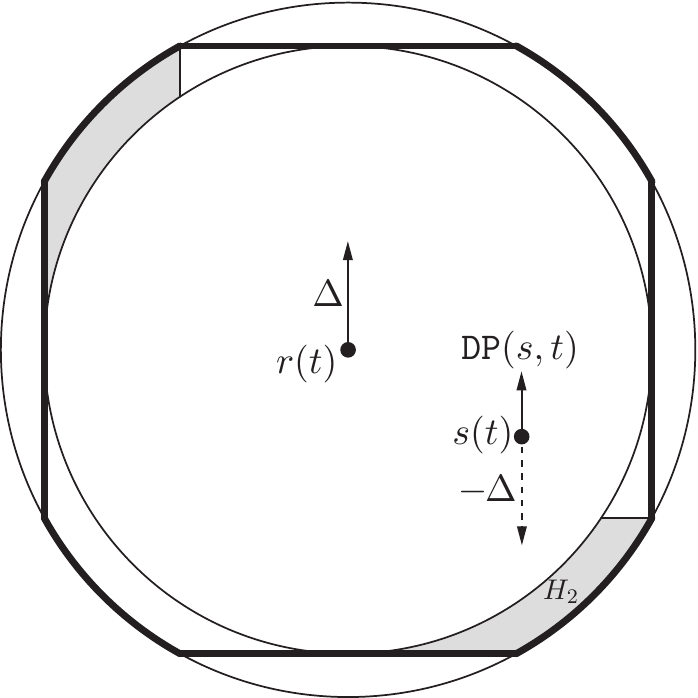}}\qquad
\subfigure[]{\label{fig:MAb}\includegraphics[scale=1]{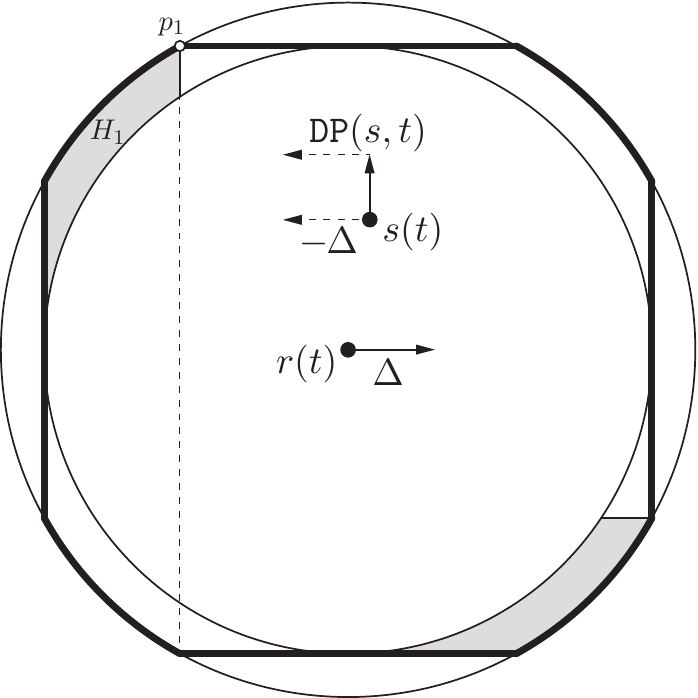}}
\caption{Proof of Lemma~\ref{lem:aw1}. The thick line is the border of $R$. In~(a), $r$ moves vertically. In~(b), $r$ moves horizontally and $s$ is to the right of $r$ at time $t$.}
\label{fig:MA}
\end{figure}

Otherwise, $\Delta.x>0$ and $\Delta.y=0$. Referring to Figure~\ref{fig:MAb}, if $s(t).x\geq r(t).x$, our claim that $\DP{s,t}-\Delta$ lies in $R(r,t)$ is trivially true, due to Proposition~\ref{prop1} and recalling that $\Delta.x\leq \rho/4$: indeed, $s(t)$ and $\DP{s,t}$ move leftward in the coordinate system of $r$ by at most $\rho/4$, hence they stay to the right of $p_1$. Moreover, $s(t)$ cannot lie in $H_1$ or else $r$ would not move rightward.

\begin{figure}[ht]
\centering
\subfigure[]{\label{fig:MAc}\includegraphics[scale=1]{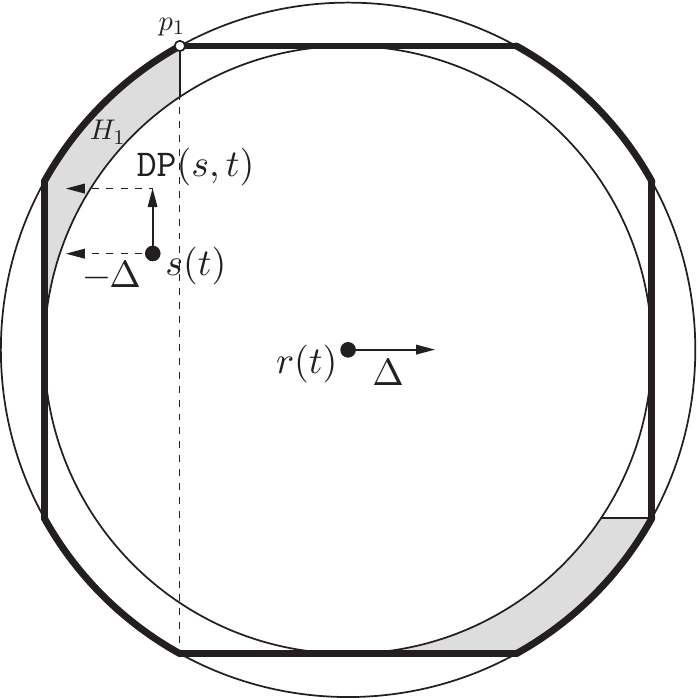}}\qquad
\subfigure[]{\label{fig:MAd}\includegraphics[scale=1]{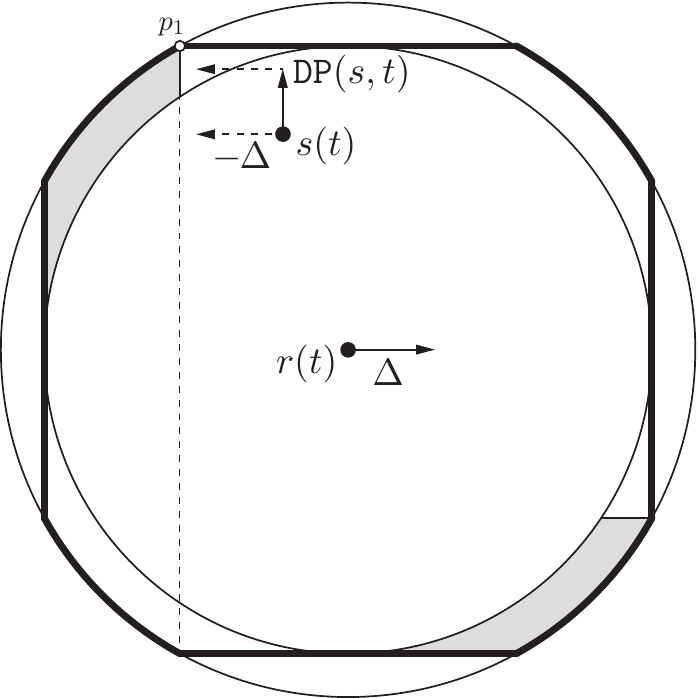}}
\caption{Proof of Lemma~\ref{lem:aw1}. The thick line is the border of $R$. In~(a), $r$ moves horizontally and $s$ is to the left of $p_1$ at time $t$. In~(b), $s$ is to the right of $p_1$ at time $t$.}
\label{fig:MA2}
\end{figure}

The only case left is that in which $s(t)$ belongs to $R\setminus H_1$ and lies to the left of $r(t)$. Recall that, according to the algorithm, $s(t)-\Delta$ belongs to $R\setminus H_1$ as well. Since $\DP{s,t}.y-s(t).y\leq \rho/4$, and due to Proposition~\ref{prop1}, it is clear that $\DP{s,t}-\Delta$ lies in $R$, provided that $s(t)-\Delta$ lies to the left of $p_1$ (see Figure~\ref{fig:MAc}). Otherwise (see Figure~\ref{fig:MAd}), if $p_1.x\leq s(t).x-\Delta.x<r(t).x$, the claim follows from the fact that $\DP{s,t}.y\leq p_1.y$ (because by assumption $\AW{r(t),\, \DP{s,t}}$), and therefore $\DP{s,t}-\Delta$ lies below $p_1$ and to its right.
\end{proof}

\begin{lemma}\label{lem:aw2}
Let two robots $r$ and $s$ be given, with $r\in\mathbb L(t_r)$ and $t_s=\Last{s,t_r}$. If $\AW{r(t_s),\, s(t_s)}$ and $\AW{r(t_r),\, s(t_r)}$, then $\AW{r(t_r),\, \DP{s,t_r}}$.
\end{lemma}
\begin{proof}
From $t_s=\Last{s,t_r}$ it follows that $t_s\leq t_r$. If $t_s=t_r$, then $s\in\mathbb L(t_r)$ and, due to Observation~\ref{obs:AWR}, $\AW{\DP{s,t_r},\, r(t_r)}$, which is our claim. So let us assume that $t_s<t_r$. If $\DP{s,t_s}=s(t_s)$, there is nothing to prove, because in this case $\DP{s,t_r} = \DP{s,t_s} = s(t_s) = s(t_r)$. So we may assume that $s$ moves strictly vertically (cf.\ Observation~\ref{obs:sym}), and therefore $\DP{s,t_s}.x=s(t_s).x$ and $s(t_s).y< \DP{s,t_s}.y\leq s(t_s).y+\rho/4$. Let $\Delta=\DP{s,t_s}-s(t_s)$. Also observe that, by definition of $t_s$, $\DP{s,t_r}=\DP{s,t_s}$.

We reason by considering the ``point of view'' of robot $r$. Let $\Delta'=r(t_r)-r(t_s)$. Hence $\DP{s,t_r}-r(t_r)=\DP{s,t_s}-\Delta'-r(t_s)$. In other terms, as a consequence of $r$ moving upward and rightward (by $\Delta'$) between $t_s$ and $t_r$, $\DP{s,t}$ moves downward and leftward in the coordinate system of $r$, as $t$ varies from $t_s$ to $t_r$.

Recall that $\AW{r(t_r),\, s(t_r)}$ by hypothesis, and hence $s(t_r)\in R(r,t_r)$. If $s(t_r).y\leq r(t_r).y$, then, by Proposition~\ref{prop1}, $\DP{s,t_r}\in R(r,t_r)$, as desired. Therefore, assume that $s(t_r).y> r(t_r).y$. This also implies that $\DP{s,t}.y> r(t).y$ for all $t\in[t_s,t_r]$.
Note that $|s(t).x-r(t).x|\leq V-\rho$, for every $t\in[t_s,t_r]$. Indeed, the inequality holds at times $t_s$ and $t_r$ by the hypotheses of the lemma, and moreover $s(t).x$ is independent of $t\in[t_s,t_r]$, while $r(t).x$ may only increase.

Let $t'=\First{r,t_s}$. We claim that both $s(t')$ and $\DP{s,t'}$ belong to $R(r,t')$. Assume first that $r$ moves upward (or stays still) between $t_s$ and $t'$. Then, by Observation~\ref{obs:AWR} and the convexity of $R$, the segment with endpoints $s(t_s)$ and $\DP{s,t_s}$ lies in $R(r,t_s)$. If such a segment moves downward in the coordinate system of $r$ (as a consequence of $r$ moving upward), and, at time $t'$, $s$ lies strictly below $R(r)$, this implies that $\DP{s,t'}.y$ cannot be greater than $r(t').y$, due to Proposition~\ref{prop1} (recall that $\DP{s,t'}.y-s(t').y\leq \rho/4$). This contradicts the assumption on $\DP{s,t'}.y$ made in the previous paragraph.

So, let $r$ move rightward, and let $r(t')=r(t_s)+\Delta''$, with $0<\Delta''.x\leq \rho/4$. Hence, if $s(t_s).x\geq r(t_s).x$, our claim is once again easily proven. Indeed, by Observation~\ref{obs:AWR}, $\DP{s,t_s}$ lies in $R(r,t_s)$, as well as $s(t_s)$. Then, by Proposition~\ref{prop1}, these two points cannot move outside of $R(r)$ as $r$ moves rightward by at most $\rho/4$, provided that $s(t_s).x=\DP{s,t_s}.x\geq r(t_s).x$. So, let us assume that $s(t_s).x< r(t_s).x$.

Since by hypothesis $s$ moves strictly upward based on a \Look performed at time $t_s$, it means that $r(t_s)\notin H_2(s,t_s)$. Equivalently, $s(t_s)$ does not belong to the region symmetric to $H_2(r,t_s)$ with respect to $r(t_s)$, which we denote by $-H_2(r,t_s)$ (see Figure~\ref{fig:MAia}). As a consequence of the algorithm (in particular, by Rule~\ref{rule5} of Section~\ref{sec:algo}), $s$ does not compute a destination point that would make $r$ enter the region $H_2$. Equivalently, in $r$'s coordinate system, $\DP{s,t_s}\notin -H_2(r,t_s)$.

\begin{figure}[ht]
\centering
\subfigure[]{\label{fig:MAia}\includegraphics[scale=1]{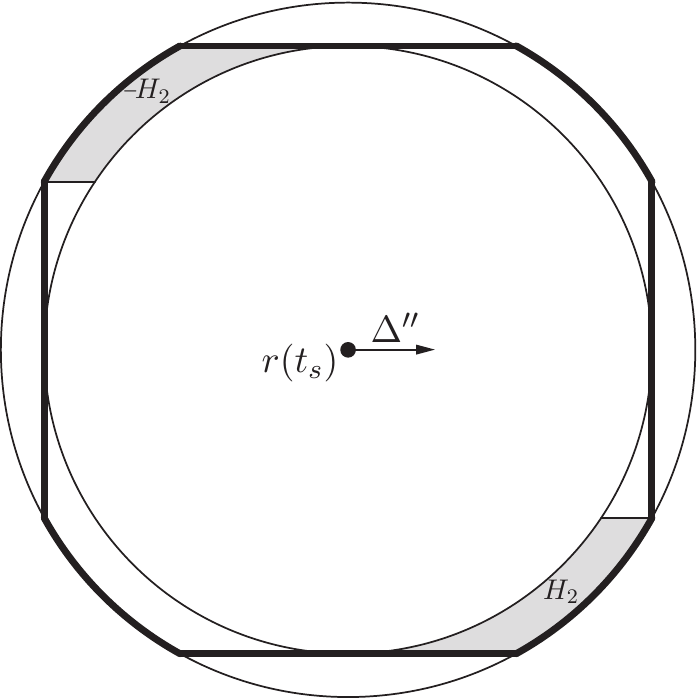}}\qquad
\subfigure[]{\label{fig:MAib}\includegraphics[scale=1.25]{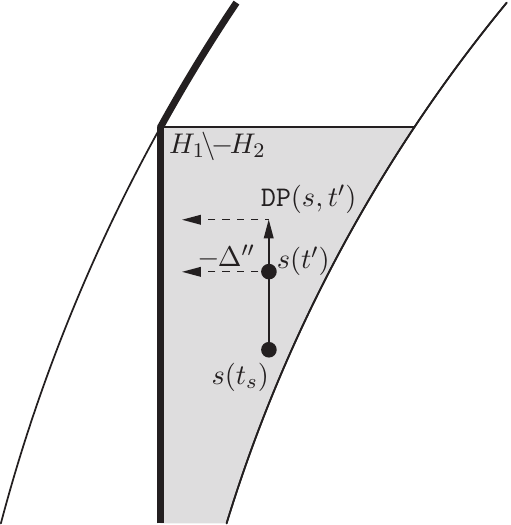}}
\caption{Proof of Lemma~\ref{lem:aw2}. In (a), the gray area in the upper-left corner is $-H_2(r,t_s)$. In (b), a detail of the set difference $H_1\setminus -H_2$ is shown.}
\label{fig:MAi}
\end{figure}

In particular, as illustrated in Figure~~\ref{fig:MAib}, if $s(t_s)\in H_1(r,t_s)\setminus -H_2(r,t_s)$, then also $\DP{s,t_s}\in H_1(r,t_s)\setminus -H_2(r,t_s)$. Hence both $s(t')$ and $\DP{s,t'}$ belong to $R(r,t')$ (recall that $|s(t').x-r(t').x|\leq V-\rho$).

Suppose now that $s(t_s)\in D_2(r,t_s)$. Note that, as a consequence of the algorithm (again, by Rule~\ref{rule5} of Section~\ref{sec:algo}), $\DP{s,t_s}.y\leq r(t_s).y+V-\rho$. Additionally, $s$ has to move by more than $\rho/2$ in the coordinate system of $r$ in order to cross the boundary of $D_1(r)$. But $\norm{\Delta+\Delta''}_2\leq \rho/4+\rho/4=\rho/2$. As a consequence,  both $s(t')$ and $\DP{s,t'}$ still belong to $D_1(r,t')$, and therefore also to $R(r,t')$ (note that we already proved that $|s(t').x-r(t').x|\leq V-\rho$).

The only case left is when $s(t_s)$ lies in the lower-left area bounded by $R(r,t_s)$ and $D_2(r,t_s)$. By Proposition~\ref{prop1} and because $\DP{s,t_s}.y-s(t_s).y\leq \rho/4$, $\DP{s,t'}$ certainly lies in $R(r,t')$. However, we also know that $\DP{s,t'}.y>r(t').y$, and that $\DP{s,t'}.y-s(t').y\leq \rho/4$. Hence, again by Proposition~\ref{prop1}, $s(t')$ must lie in $R(r,t')$ as well.

Now our claim is proven. If $t'=t_r$, we are done. Otherwise, we apply Lemma~\ref{lem:aw1} by setting $t:=t'$. As a result, $\AW{\DP{r,t'},\,s(t')}$ and $\AW{\DP{r,t'},\,\DP{s,t'}}$. Let $t''=\First{r,t'}$. By the convexity of $R$, it follows that both $s(t'')$ and $\DP{s,t''}$ belong to $R(r,t'')$ (recall that $\DP{s,t}$ does not depend on $t\in[t_s,t_r]$). If $t''=t_r$, we are done. Otherwise, we keep applying Lemma~\ref{lem:aw1} (with $t:=t''$, etc.) and repeating the previous reasoning, until we prove that $\DP{s,t_r}\in R(r,t_r)$, which concludes the proof.
\end{proof}

Now we are ready to give the full definition of \emph{mutual awareness} and the related graph.

\begin{definition}[Mutual Awareness]\label{def:ma}
Two distinct robots $r$ and $s$ are {\em mutually aware} at time $t$ if both conditions hold:
\begin{enumerate}
\item $\AW{r(t_r),\, s(t_r)}$, with $t_r=\Last{r,t}$, and
\item $\AW{r(t_s),\, s(t_s)}$, with $t_s=\Last{s,t}$.
\end{enumerate}
\end{definition}

\begin{definition}[Mutual Awareness Graph]
The {\em mutual awareness graph} at time $t\geq 0$ is the graph $\widetilde G(t) = (\setRobots, E(t))$ such that, for any two distinct robots $r$ and $s$, $\{r,s\} \in E(t)$ if and only if $r$ and $s$ are mutually aware at time $t$.
\end{definition}

We recall that $D=V-\sigma$, with $\sigma>0$ arbitrary small. By definition of ${\it Last}$ and of mutual awareness, we have the following.

\begin{observation}
\label{obs:mutual-basis} All the pairs of robots that are
at (Euclidean) distance not greater than $D$ from each other at time $t=0$ are initially mutually aware. Hence $J\subseteq \widetilde G(0)$, and therefore $\widetilde G(0)$ is connected.\qed
\end{observation}

In the following lemma, we will prove that any two robots that are mutually aware at some point keep being so during the entire execution.

\begin{lemma}
If robots $r$ and $s$ are mutually aware at time $t$, they are mutually aware at any time $t'\geq t$.
\label{lem:aware}
\end{lemma}
\begin{proof}
Let $(t_i)_{i\geq 0}$ be the strictly increasing sequence of time instants at which either $r$ or $s$ executes a \Look; if both $r$ and $s$ execute a \Look simultaneously, such a time instant appears only once in the sequence. Without loss of generality, we may assume that $r$ and $s$ first become mutually aware at time $t_m$, when $r$ enters a \Look phase.

We will prove by induction that, for all $i\geq m$, the following conditions hold:
\begin{enumerate}
\item $\AW{r(t_i),\, s(t_i)}$,
\item $\AW{\DP{r,t_i},\, s(t_i)}$,
\item $\AW{r(t_i),\, \DP{s,t_i}}$,
\end{enumerate}
which will clearly imply our claim (Condition~1 actually suffices).

Let $i=m$, and observe that Condition~1 holds by definition of mutual awareness. Moreover, by Lemma~\ref{lem:aw2} with $t_r:=t_m$, Condition~3 holds, too. Finally, Condition~2 is implied by Conditions~1 and~3 and by Lemma~\ref{lem:aw1} with $t:=t_m$.

Suppose now that $i>m$, and let the three conditions hold at every time $t_j$ with $m\leq j\leq i-1$. Without loss of generality, we may assume that $r\in\mathbb L(t_{i-1})$ (if $s\in\mathbb L(t_{i-1})$, we just exchange $r$ and $s$ in our proof). By Conditions~1 and~3 on $t_{i-1}$, we have
$$\AW{r(t_{i-1}),\, s(t_{i-1})}\mbox{ and}$$
$$\AW{r(t_{i-1}),\, \DP{s,t_{i-1}}}.$$
By Lemma~\ref{lem:aw1} with $t:=t_{i-1}$, we have also $\AW{\DP{r,t_{i-1}},\, s(t_{i-1})}$ and $\AW{\DP{r,t_{i-1}},\, \DP{s,t_{i-1}}}$. These are equivalent, respectively, to
$$\AW{r(t_{i-1}),\, s(t_{i-1})-\DP{r,t_{i-1}}+r(t_{i-1})}\mbox{ and}$$
$$\AW{r(t_{i-1}),\, \DP{s,t_{i-1}}-\DP{r,t_{i-1}}+r(t_{i-1})}.$$
Collectively, $s(t_{i-1})$, $\DP{s,t_{i-1}}$, $s(t_{i-1})-\DP{r,t_{i-1}}+r(t_{i-1})$ and $\DP{s,t_{i-1}}-\DP{r,t_{i-1}}+r(t_{i-1})$ are four points whose convex hull $C$ is either a rectangle or a segment (depending if $r$ and $s$ move orthogonally or parallel to each other between $t_{i-1}$ and $t_i$). Because the vertices of $C$ are contained in $R(r,t_{i-1})$ (cf.~the definition of $R$ in the algorithm), and because $R$ is convex, $C$ is entirely contained in $R(r,t_{i-1})$ (refer to Figure~\ref{fig:algob}).

Moreover, $r(t_i)$ (resp.\ $s(t_i)$) lies on the segment with endpoints in $r(t_{i-1})$ and $\DP{r,t_{i-1}}$ (resp.\ $s(t_{i-1})$ and $\DP{s,t_{i-1}}$). Let $r(t_i)=r(t_{i-1})+\Delta_r$ and $s(t_i)=s(t_{i-1})+\Delta_s$. So, the point $s(t_{i-1})+\Delta_s-\Delta_r$ belongs to $C$, and therefore to $R(r,t_{i-1})$. In other terms,
$$\AW{r(t_{i-1}),\, s(t_{i-1})+\Delta_s-\Delta_r},$$
which is equivalent to $\AW{r(t_{i-1})+\Delta_r,\, s(t_{i-1})+\Delta_s}$, and to $\AW{r(t_{i}),\, s(t_{i})}$. Hence Condition~1 holds at $t_i$.

Once again, without loss of generality, we may assume that $r\in\mathbb L(t_i)$. Then, Condition~3 at $t_i$ follows from Condition~1 and Lemma~\ref{lem:aw2} with $t_r:=t_i$. Condition~2, on the other hand, follows from Conditions~1 and~3, and from Lemma~\ref{lem:aw1} with $t:=t_i$.
\end{proof}

\begin{corollary}
\label{cor:connectaware}
$\widetilde G(t)$ is connected at any time $t\geq 0$.
\end{corollary}
\begin{proof}
$\widetilde G(0)$ is connected by Observation~\ref{obs:mutual-basis}. By Lemma~\ref{lem:aware}, $\widetilde G(0)$ is a subgraph of $\widetilde G(t)$, and therefore $\widetilde G(t)$ is connected.
\end{proof}

\begin{corollary}
\label{cor:connectvisible}
$\widetilde G(t)\subseteq G(t)$, and therefore $G(t)$ is connected at any time $t\geq 0$.
\end{corollary}
\begin{proof}
Suppose that robots $r$ and $s$ are mutually aware at time $t$. Then, by Lemma~\ref{lem:aware}, they are mutually aware at any time after $t$, regardless of the scheduler's choices. Moreover, observe that the proof of Lemma~\ref{lem:aware} goes through even if the scheduler can stop the robots before they have moved by at least $\delta$ (recall that the fairness assumption of our robot model normally forbids the scheduler to interrupt a robot's \Move phase before it has moved by at least $\delta$).

Let us therefore modify the execution of $r$ and $s$, and let the scheduler interrupt their \Move phase precisely at time $t$, regardless of how much they have actually moved during that phase. By the previous observations, $r$ and $s$ are mutually aware at time $t'=\First{r,t}$, and additionally $r(t)=r(t')$ and $s(t)=s(t')$. Hence, by definition of mutual awareness, $r$ and $s$ are at (Euclidean) distance not grater than $V-\rho/2$ at time $t'$, and therefore also at time $t$.

This implies that $\widetilde G(t)\subseteq G(t)$, and hence that $G(t)$ is connected, by Corollary~\ref{cor:connectaware}.
\end{proof}

\subsection{Collision Avoidance\label{sec:CA}}

In this section, we will prove that no collision occurs during the execution of the algorithm. 

\begin{lemma}
\label{lem:collision}
No collision ever occurs between any pair of robots during the execution of the algorithm.
\end{lemma}
\begin{proof}
Let us assume by contradiction that two distinct robots $r$ and $s$ collide during their execution. Because $r(t)$ and $s(t)$ are continuous functions, there exists a minimum time instant $t>0$ at which $r(t)=s(t)=p$. At least one robot, say $r$, must make a strictly positive movement toward $p$, at some point. Let $t'<t$ be the last time at which $r$ performs a \Look phase such that $r(t')\neq p$. Recall that, by Observation~\ref{obs:1}, $r$ and $s$ move either upward or rightward at each move. Without loss of generality (cf.~Observation~\ref{obs:sym}), let us assume that $r$ moves strictly rightward between $t'$ and $t$. Then, by Observation~\ref{obs:3}, $0< p.x-r(t').x\leq V/16$. Several cases arise.

If $s(0)=p$, then $s(t')=p\in\Qii(r,t')$, which is a contradiction because, by the algorithm (specifically, by Rule~\ref{rule2} of Section~\ref{sec:algo}), $\DP{r,t'}.x$ must be less than the $x$-coordinate of every robot in $\Qii(r,t')$, and therefore $r$ cannot be found in $p$ at time $t$.

If $s(0)\neq p$, then $s$ performs at least one positive movement to reach $p$. Let $t''<t$ be the last time at which $s$ performs a \Look phase such that $s(t'')\neq p$. By symmetry between $r$ and $s$, we may assume that $t''\leq t'<t$.

Suppose that $s$ moves strictly upward between $t''$ and $t$ (see Figure~\ref{fig:CAa}). Hence $0< p.y-s(t'').y\leq V/16$. Because $t''\leq t'$, it follows that $s(t')\in\Qii(r,t')$, which contradicts the fact that $r$ reaches $p$ in the next move.

\begin{figure}[ht]
\centering
\subfigure[]{\label{fig:CAa}\includegraphics[scale=1]{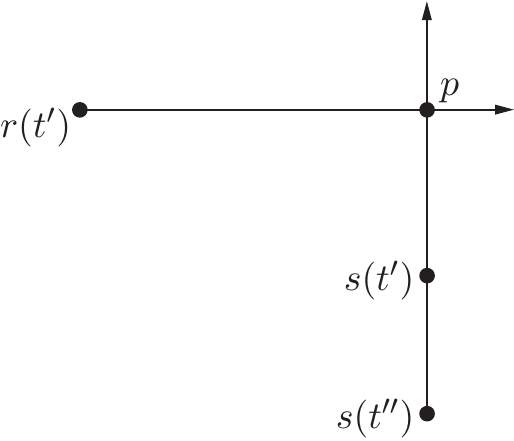}}\qquad
\subfigure[]{\label{fig:CAb}\includegraphics[scale=1]{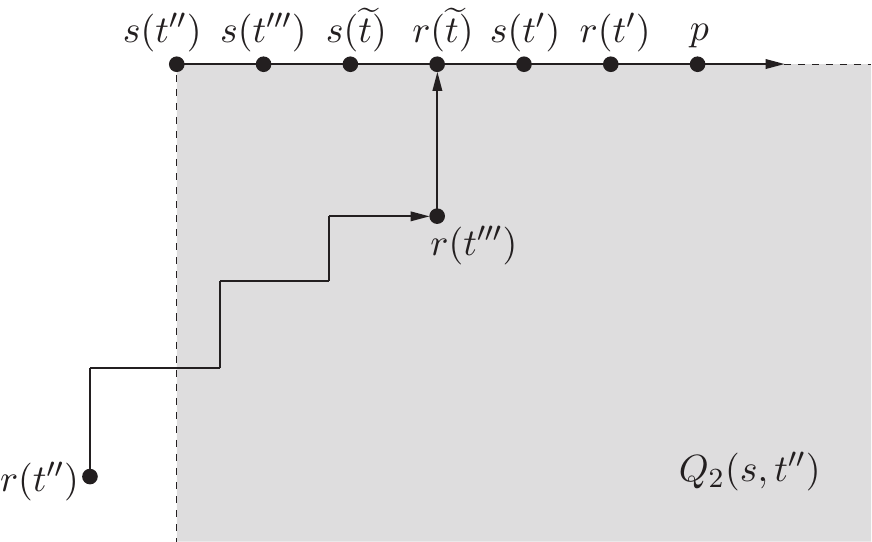}}
\caption{Proof of Lemma~\ref{lem:collision}. In~(a), $s$ moves upward between $t''$ and $t$. In~(b), $s$ moves rightward.}
\label{fig:CA}
\end{figure}

Otherwise, $s$ moves strictly rightward between $t''$ and $t$. Since $t''\leq t'$, it follows that $s(t'').y=s(t').y=p.y$ (see Figure~\ref{fig:CAb}). $s(t')$ cannot lie to the right of $r(t')$, otherwise it would be in $\Qii(r,t')$, yielding a contradiction with the algorithm (Rule~\ref{rule2} of Section~\ref{sec:algo}). Hence $s(t'').x\leq s(t').x\leq r(t').x< p.x$. We claim that $r(t'').y<s(t'').y$. Indeed, suppose by contradiction that $r(t'').y=s(t'').y$. If $r(t'').x>s(t'').x$, then $r(t'')\in \Qii(s,t'')$ and $s$ computes a destination point that is not to the left of $r$, which again contradicts Rule~\ref{rule2} of the algorithm. Otherwise $r(t'').x\leq s(t'').x$, which implies that $r$ and $s$ collide between $t''$ and $t'$, contradicting the minimality of $t$.

Because $r.y<p.y$ at time $t''$ and $r.y=p.y$ at time $t'$, there is a time $\widetilde t\in (t'',t']$ at which $r.y$ first becomes equal to $p.y$. Note that $r(\widetilde t).x>s(\widetilde t).x$, otherwise $r$ and $s$ would collide between $\widetilde t$ and $t'$. Hence $r(\widetilde t)\in \Qii(s,\widetilde t)$. Note also that $r(t'')\not\in \Qii(s,t'')$, because $\DP{s,t''}.x\geq r(t'').x$. Since each move covers at most $V/16$, $r$ performs more than one move between $t''$ and $\widetilde t$: if $r$ enters $\Qii(s)$ for the last time from below, it must move vertically more than once; if $r$ enters $\Qii(s)$ for the last time from the left, then it must turn upwards at some point (refer to Figure~\ref{fig:CAb}). More precisely, $r$ performs at least one \Look phase in $[t'',\widetilde t)$, the last of which at time $t'''$, and $r$ moves strictly upward between $t'''$ and $\widetilde t$. Then
$$s(t'').x \leq s(t''').x \leq s(\widetilde t).x < r(\widetilde t).x = r(t''').x.$$
It follows that $0<r(t''').x-s(t''').x\leq V/16$. Moreover, $0<s(t''').y-r(t''').y\leq V/16$, hence $s(t''')\in\Qi(r,t''')$. This contradicts Rule~\ref{rule2} of the algorithm, because $\DP{r,t'''}.y\geq s(t''').y$.
\end{proof}

\subsection{Convergence and Termination\label{sec:T}}

In this final section, we will prove that the robots will converge to the same limit point (Lemma~\ref{lem:convergence}), and then finally that our \NearG algorithm is correct (Theorem~\ref{thm:correct}).

Let $\ell$ be the point having the $x$-coordinate of the rightmost point in \IC, and the $y$-coordinate of the topmost point in \IC. That is,
$$\ell=\left(\max_{r\in\setRobots}\left\{r(0).x\right\},\max_{r\in\setRobots}\left\{r(0).y\right\}\right).$$

\begin{lemma}
\label{lem:convergence}
If no robot ever terminates its execution, then all robots converge towards point $\ell$.
\end{lemma}
\begin{proof}
Let an execution of the robot set \setRobots be fixed, in which no robot ever terminates. By Observation~\ref{obs:1}, the movement of each robot is monotonically increasing with respect to both the $x$-coordinate and the $y$-coordinate. Also, at any time, each robot's coordinates are bounded from above by the coordinates of $\ell$. It follows that each robot $r$ converges towards a point, denoted by $\LIM(r)$, such that $\LIM(r).x\leq \ell.x$ and $\LIM(r).y\leq \ell.y$.

If all robots have the same convergence point, then this point must be $\ell$, because there is a robot whose $x$-coordinate is constantly $\ell.x$ and a (possibly distinct) robot whose $y$-coordinate is constantly $\ell.y$. Hence, in this case the lemma follows. Thus, let us assume that there is more than one convergence point. Let $\lambda\in\mathbb R^+$ be any positive number such that:
\begin{itemize}
\item $\lambda\leq\LIM(r).x-\LIM(s).x$ for every $r,s\in\setRobots$ with $\LIM(r).x>\LIM(s).x$;
\item $\lambda\leq\LIM(r).y-\LIM(s).y$ for every $r,s\in\setRobots$ with $\LIM(r).y>\LIM(s).y$;
\item $\lambda\leq V-\dist(\LIM(r),\, \LIM(s))$ for every $r,s\in\setRobots$ with $\dist(\LIM(r),\, \LIM(s))<V$;
\item $\lambda \leq \min\left\{\rho/2,\,\delta\right\}$.
\end{itemize}
Because \setRobots is a finite set, there is a time $t_0$ at which, for every $r\in\setRobots$,
$$\dist(r(\Last{r,t_0}),\, \LIM(r))<\lambda/3.$$

By definition of $\lambda$, if $\dist(\LIM(r),\, \LIM(s))< V$, then $\dist(r(t),\, s(t))<V$ for all $t\geq t_0$. On the other hand, if $\AW{r(t),\, s(t)}$ for some $t\geq t_0$, then in particular $\dist(r(t),\, s(t))\leq V-\rho/2$, and therefore $\dist(\LIM(r),\, \LIM(s))<V$.

Let us choose $t_1>t_0$ such that every robot in \setRobots executes at least one complete cycle between $t_0$ and $t_1$ (i.e., from a \Look phase to the next). We further assume that, for every $r\in\setRobots$, if $r(t_0).x<\LIM(r).x$ (resp.\ $r(t_0).y<\LIM(r).y$), then in at least one such cycle (i.e., executed between $t_0$ and $t_1$) $r$ moves strictly rightward (resp.\ upward). Note that we can make this assumption because $\LIM(r).x$ must be approached indefinitely by $r.x$, and therefore, if $r(t_0).x<\LIM(r).x$, then $r$ must make a rightward move at some point after $t_0$ (and similarly for $y$-coordinates and upward moves).

Let $a$ be the lowest among the leftmost convergence points of the robots in \setRobots, and let $\mathcal A\subset \setRobots$ be the set of robots that converge towards $a$.

Suppose first that there exists some robot $s\in\setRobots\setminus \mathcal A$ converging to $b\neq a$, such that $\dist(a,b)<V$ and $b.x>a.x$. Let $r$ be any rightmost robot of $\mathcal A$ at time $t_1$. As observed three paragraphs above, $r$ and $s$ can see each other at any time since $t_0$.

\begin{figure}[ht]
\centering
\subfigure[]{\label{fig:CTa}\includegraphics[scale=0.6]{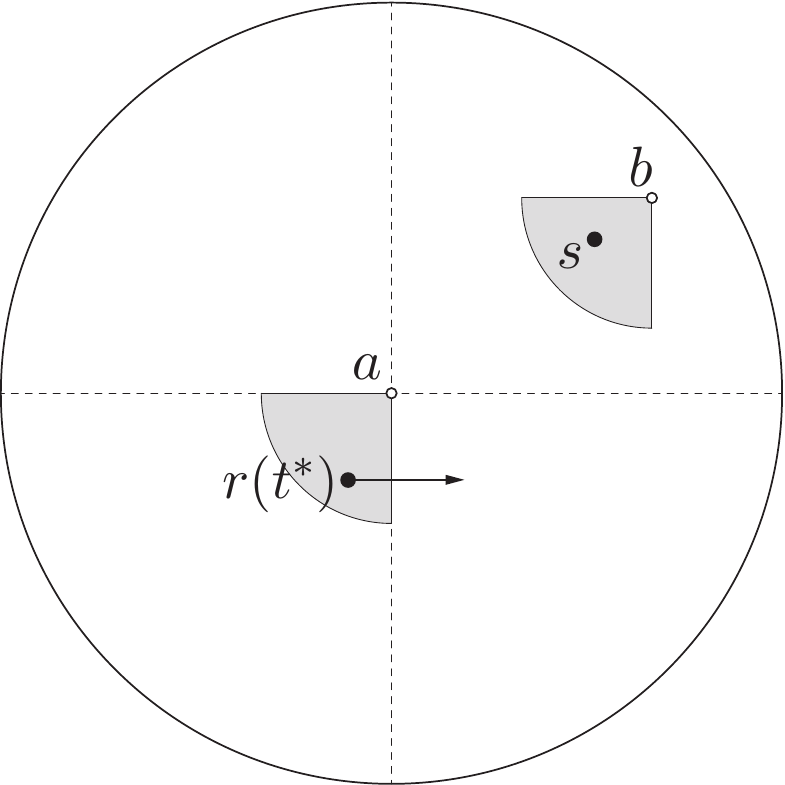}}\quad
\subfigure[]{\label{fig:CTb}\includegraphics[scale=0.6]{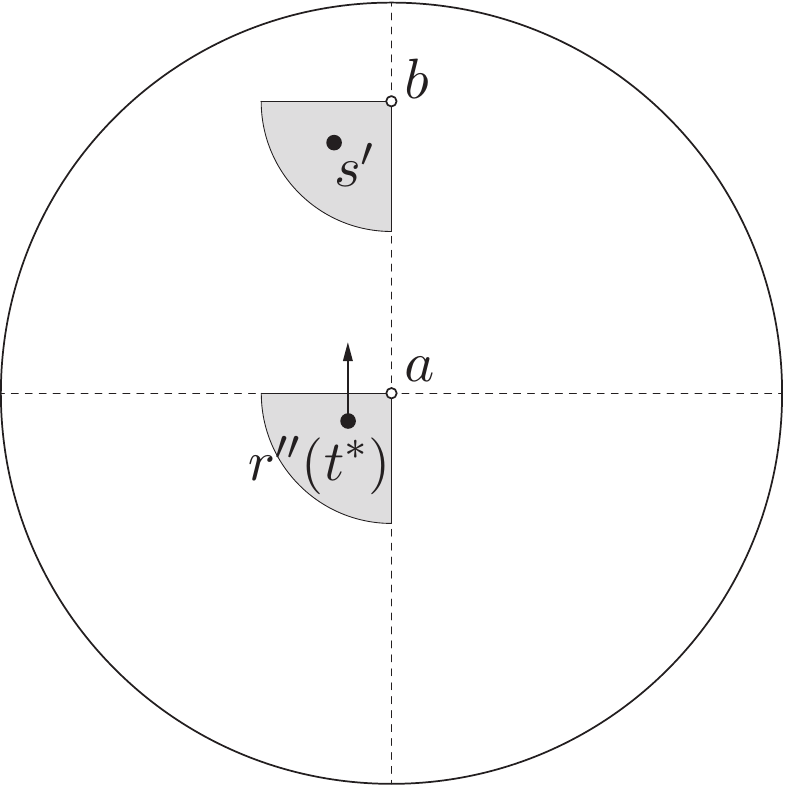}}\quad
\subfigure[]{\label{fig:CTc}\includegraphics[scale=0.6]{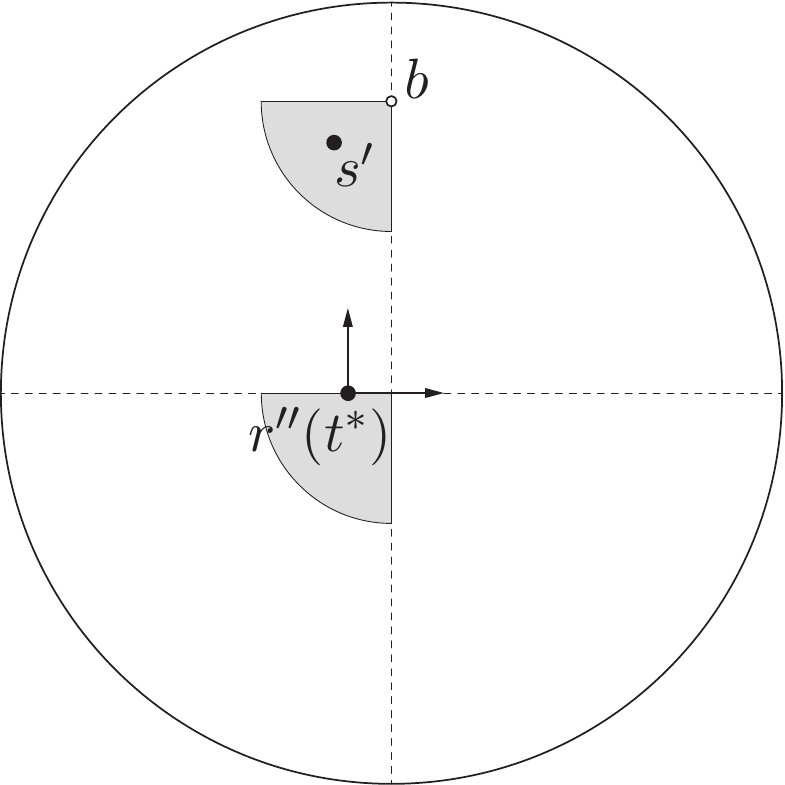}}
\caption{Proof of Lemma~\ref{lem:convergence}. In~(a), $b$ lies strictly to the right of $a$. In~(b), $a.x=b.x$ and $r''(t_1).y<a.y$. In~(c), $a.x=b.x$ and $r''(t_1).y=a.y$.}
\label{fig:CT}
\end{figure}

If $r.x<a.x$ then, by construction, there exists a time $t^*\in[t_0,t_1]$ at which $r$ performs a \Look phase, such that $r(t^*).x<\DP{r,t^*}.x$ and $r(\First{r,t^*}).x=r(t_1).x$ (see Figure~\ref{fig:CTa}). According to the algorithm (specifically, by Rule~\ref{rule2} of Section~\ref{sec:algo}), if $r$ is able to compute such a destination point, it means that no robot of $\mathcal A$ lies in $\Qii(r)$ at time $t^*$. Therefore, by definition of $\lambda$ and by construction, every robot in $\Qii(r,t^*)$ has an $x$-coordinate that is greater than $a.x+2\lambda/3$. Because $s(t^*).x>a.x+2\lambda/3$ as well, it follows that $\DP{r,t^*}.x>r(t^*).x+\lambda/3$ (observe that no robot in $\Qi(r,t^*)$ can prevent $r$ from moving rightward by at least $\rho/4>\lambda/3$, due to Proposition~\ref{prop1}). Additionally, $\lambda/3 <\delta$, hence $r$ actually moves by more than $\lambda/3$. But $r(t^*).x+\lambda/3 > a.x$, contradicting the fact that $r.x$ monotonically converges to $a.x$.

Otherwise, $r.x=a.x$ holds. Then, let $t^*=\First{r,t_1}$. At time $t^*$, $r$ sees no robot $q$ with $r(t^*).x<q(t^*).x\leq r(t^*).x+2\lambda/3$. Moreover, $r$ sees $s$, and $s(t^*).x>r(t^*).x+2\lambda/3$. Hence, $\DP{r,t^*}$ has distance greater than $\lambda/3$ from $r(t^*)$, and $r$ actually moves rightward or upward by more than $\lambda/3 <\delta$. When $r$ is done moving, either $r.x>a.x$ or $r.y>a.y$, contradicting the fact that $\LIM(r)=a$.

Suppose now that there is no limit point $b\neq a$ such that $\dist(a,b)<V$ and $b.x>a.x$. By Corollary~\ref{cor:connectaware}, $\widetilde G(t_0)$ is connected, hence there exist robots $r'\in\mathcal A$ and $s'\in\setRobots\setminus\mathcal A$ that are mutually aware at time $t_0$ (and also at any time $t\geq t_0$, by Lemma~\ref{lem:aware}). Let $b=\LIM(s')$. Then, either $\dist(a,b)\geq V$ or $a.x=b.x$ (recall that $a$ is a leftmost convergence point). However, observe that if $\dist(a,b)\geq V$, then $r'$ and $s'$ cannot be mutually aware at any time $t\geq t_0$, because $\dist(r'(t),s'(t))>V-\rho/2$. Therefore, $a.x=b.x$ and $a.y<b.y<a.y+V$.

Let $r''$ be any topmost robot of $\mathcal A$ at time $t_1$. By Corollary~\ref{cor:connectvisible}, $r'$ sees $s'$ at any time $t\geq t_0$. Then, by definition of $\lambda$ and $t_0$, also $r''$ sees $s'$ at any time $t\geq t_0$.

Suppose first that $r''(t_1).y<a.y$ (see Figure~\ref{fig:CTb}). By definition of $t_1$, there exists a time $t^*\in[t_0,t_1]$ at which $r''$ performs a \Look phase, such that $r''(t^*).y<\DP{r'',t^*}.y$ and $r''(\First{r'',t^*}).y=r''(t_1).y$. According to the algorithm (specifically, by Rule~\ref{rule2} of Section~\ref{sec:algo}), if $r''$ is able to compute such a destination point, it means that no robot of $\mathcal A$ lies in $\Qi(r'')$ at time $t^*$. Therefore, by definition of $\lambda$ and by construction, every robot in $\Qi(r'',t^*)$ has a $y$-coordinate that is greater than $a.y+2\lambda/3$. On the other hand, $r''$ sees no robot $q$ with $q(t^*).y<r''(t^*).y$ and $\dist(r''(t^*),q(t^*))\leq V-\rho/2$, hence $r''$ is able to move upward by more than $\lambda/3$. But indeed, $r''$ does see $s'$, and $s'(t^*).y>a.y+2\lambda/3$, hence $\DP{r'',t^*}.y>r''(t^*).y+\lambda/3$. Once again, this contradicts the fact that $\LIM(r'')=a$.

Finally, suppose that $r''(t_1).y=a.y$ (see Figure~\ref{fig:CTc}), and let $t^*=\First{r'',t_1}$. At time $t^*$, $r''$ sees no robot $q$ in $\Qi(r'',t^*)$ with $r''(t^*).y<q(t^*).y\leq r'(t^*).y+2\lambda/3$. Similarly to the previous paragraph's case, no robot below $r''$ can prevent $r''$ from moving upward, and the presence of $s'$ makes $r''$ compute a destination point that is more than $\lambda/3<\delta$ away from $r''$. Thus, $r''$ moves either upward of rightward by more than $\lambda/3$, which is in contradiction with the fact that $\LIM(r'')=a$.
\end{proof}

To prove that termination is correctly detected, we need one last lemma.

\begin{lemma}
A robot $r$ terminates its execution at time $t$ only if it sees all the robots in $\mathcal R$ at time $\Last{r,t}$.
\label{lem:lights_convergence}
\end{lemma}
\begin{proof}
Let $r$ terminate its execution at time $t$, and let $\mathcal Z$ be the set of robots that are at distance at most $\varepsilon'$ from $r$ at time $t'=\Last{r,t}$. Note that $\mathcal Z$ is not empty, because $r\in \mathcal Z$. Because $r$ terminates, it follows that every robot in $\setRobots\setminus \mathcal Z$ has distance greater than $V$ from $r$ at time $t'$.

Assume for a contradiction that $\setRobots\setminus \mathcal Z$ is not empty. By Corollary~\ref{cor:connectvisible}, $G(t')$ is connected, and therefore there are a robot $s\in\mathcal Z$ and a robot $s'\in \setRobots\setminus \mathcal Z$ such that $\dist(s(t'),\,s'(t'))\leq V-\rho/2$. Since $\dist(s'(t'),\,r(t'))>V$, then $\dist(s(t'),\,r(t'))>\rho/2$, by the triangle inequality. But $s\in \mathcal Z$, hence $\dist(s(t'),\,r(t'))\leq \varepsilon'\leq \rho/2$, which yields a contradiction.
\end{proof}

By putting together all the previous results, we obtain the following.

\begin{theorem}
\label{thm:correct}
The algorithm in Figure~\ref{algo:1} correctly solves the \NearG problem under Assumption~\ref{a:assumption1}.
\end{theorem}
\begin{proof}
Assume for a contradiction that no robot ever terminates its execution.
Due to Lemma~\ref{lem:convergence}, all the robots converge to the same point $\ell$. Therefore, when all the robots are contained in a square $Q$ with diagonal length $\varepsilon'$ and upper-right vertex $\ell$, they all see each other at distance not greater than $\varepsilon'$. From this time onward, whenever a robot executes a \Look and then a \Compute phase, it terminates, contradicting our assumption.

Hence, at least one robot $r$ will terminate its execution at some point in time $t>0$. Due to Lemma~\ref{lem:lights_convergence}, $r$ sees all the robots in $\setRobots$ at time $t'=\Last{r,t}$. This means that, at time $t'$, all the robots are within distance $\varepsilon'$ from each other. Due to Observation~\ref{obs:1} and Lemma~\ref{lem:convergence}, at any time $t''\geq t'$, all the robots are contained in a square $Q$ with diagonal length $\varepsilon'$ and upper-right vertex $\ell$. Then, by the same reasoning used in the previous paragraph, we conclude that every robot eventually terminates its execution while lying in $Q$.

By Lemma~\ref{lem:collision}, no two robots ever collide, and hence they correctly solve \NearG.
\end{proof}

\section{Conclusions\label{sec:conclusions}}

In this paper we presented the first algorithm that solves the \NearG problem using the standard Euclidean distance (in contrast with~\cite{NearGSIROCCO}) for a set of autonomous mobile robots with limited visibility. The protocol presented here is collision-free and handles termination: this allows to potentially combine our protocol with solutions to other problems designed for the unlimited visibility setting. This is achieved without assuming that the total number of robots in the system is known to the robots themselves, and without allowing them to explicitly communicate.

We remark that our algorithm can be easily modified to solve the \NearG problem in the robot models that use any $p$-norm distance as opposed to the Euclidean distance, including the infinity norm distance.

Moreover, our algorithm is perfectly symmetric with respect to the $x$- and $y$-axes. This implies that our solution works also when the robots agree only on the direction of one of the two axes, say, the $y$-axis, and not necessarily on the orientation of the other axis.

\begin{corollary}
The \NearG problem is solvable under Assumption~\ref{a:assumption1} even if the robots agree only on the direction of one of the two axes.
\end{corollary}
\begin{proof}
Suppose without loss of generality that the robots agree on the $y$-axis. Then, the following algorithm is employed: the input snapshot is first rotated clockwise by $45^\circ$, then the algorithm in Figure~\ref{algo:1} is applied to the resulting snapshot, and finally the computed destination point $dp$ is rotated counterclockwise by $45^\circ$.

Indeed, the two rotations effectively tilt the coordinate systems of all robots, in such a way that their $y$-axes become actually parallel to the line $y=x$ in the ``global'' coordinate system. This is equivalent to having the robots agree only on the positive direction of the line $y=x$, but allowing them to disagree on which is the $x$-axis and which is the $y$-axis. The algorithm in Figure~\ref{algo:1} still works because, due to Observation~\ref{obs:sym}, it is symmetric with respect to $x$- and $y$-coordinates.
\end{proof}

Therefore, under Assumption~\ref{a:assumption1}, the \NearG protocol can also be used to solve the classical gathering problem in the limited visibility scenario, when the robots have only this form of partial agreement on their local coordination systems, thus improving on~\cite{FloPSW05}, which requires total agreement on both axes and does not avoid collisions. Indeed, it is sufficient to convert the termination command in the algorithm in Figure~\ref{algo:1} with a move to point $\ell$, as defined in Section~\ref{sec:T}.

{We conjecture that no algorithm can solve \NearG with no agreement on at least one axis, and we leave this as an open problem. Another direction for future research would be to solve \NearG from any initial configuration in which the distance graph is connected, with no further assumption on the initial strong distance graph (cf.~Assumption~\ref{a:assumption1}). Again, we conjecture this problem to be unsolvable in \ASYNC; note that in this case some extra assumption is required, for instance that the total number of robots is known, or that robots are able to communicate. Finally, the more general model in which robots do not necessarily have the same visibility radius, and hence do not share a common unit distance, should be considered in conjunction with both the gathering problem and \NearG.}

\subsection*{Acknowledgments}
We would like to thank Paola Flocchini, Nicola Santoro, and Peter Widmayer, who contributed to the writing of this paper by sharing their ideas. We also thank the anonymous reviewers for precious comments that helped us improve the readability of this paper.

\bibliographystyle{plain}
\bibliography{NearG}

\begin{thebibliography}{10}

\bibitem{AGM13}
C.~Agathangelou, C.~Georgiou, and M.~Mavronicolas.
\newblock A distributed algorithm for gathering many fat mobile robots in the
  plane.
\newblock In {\em Proceedings of the 32nd Annual ACM Symposium on Principles of
  Distributed Computing (PODC)}, pages 250--259, 2013.

\bibitem{andoi}
H.~Ando, Y.~Oasa, I.~Suzuki, and M.~Yamashita.
\newblock A distributed memoryless point convergence algorithm for mobile
  robots with limited visibility.
\newblock {\em IEEE Transaction on Robotics and Automation}, 15(5):818--828,
  1999.

\bibitem{cie04}
M.~Cieliebak.
\newblock Gathering non-oblivious mobile robots.
\newblock In {\em 6th Latin American Conference on Theoretical Informatics
  (LATIN)}, LNCS 2976, pages 577--588, 2004.

\bibitem{CohP05}
R.~Cohen and D.~Peleg.
\newblock Convergence properties of the gravitational algorithms in
  asynchronous robots systems.
\newblock {\em SIAM Journal on Computing}, 34(6):1516--1528, 2005.

\bibitem{CLDFH11}
A.~Cord-Landwehr, B.~Degener, M.~Fischer, M.~H{\"u}llmann, B.~Kempkes,
  A.~Klaas, P.~Kling, S.~Kurras, M.~M{\"a}rtens, F.~Meyer auf~der Heide,
  C.~Raupach, K.~Swierkot, D.~Warner, C.~Weddemann, and D.~Wonisch.
\newblock Collision-less gathering of robots with an extent.
\newblock In {\em Proceedings of the 37th International Conference on Current
  Trends in Theory and Practice of Computer Science (SOFSEM)}, pages 178--189,
  2011.

\bibitem{pelc09}
J.~Czyzowicz, L.~Gasieniec, and A.~Pelc.
\newblock Gathering few fat mobile robots in the plane.
\newblock {\em Theoretical Computer Science}, 410(6--7):481--499, 2009.

\bibitem{floccICDCS2012}
S.~Das, P.~Flocchini, G.~Prencipe, N.~Santoro, and M.~Yamashita.
\newblock The power of lights: synchronizing asynchronoys robots using visibile
  bits.
\newblock In {\em Proceedings of the 32nd International Conference on
  Distributed Computing Systems (ICDCS)}, pages 506--515, 2012.

\bibitem{DefS08}
X.~D\'efago and S.~Souissi.
\newblock Non-uniform circle formation algorithm for oblivious mobile robots
  with convergence toward uniformity.
\newblock {\em Theoretical Computer Science}, 396(1--3):97--112, 2008.

\bibitem{DieLP08}
Y.~Dieudonn\'e, O.~Labbani-Igbida, and F.~Petit.
\newblock Circle formation of weak mobile robots.
\newblock {\em ACM Transactions on Autonomous and Adaptive Systems}, 3(4),
  2008.

\bibitem{petitLeaderPF}
Y.~Dieudonn\'e, F.~Petit, and V.~Villain.
\newblock Leader election problem versus pattern formation problem.
\newblock In {\em Proceedings of the 24th International Symposium on
  Distributed Computing (DISC)}, LNCS 6343, pages 267--281, 2010.

\bibitem{efrima07}
A.~Efrima and D.~Peleg.
\newblock Distributed models and algorithms for mobile robot systems.
\newblock In {\em Proceedings of the 33rd International Conference on Current
  Trends in Theory and Practice of Computer Science (SOFSEM)}, LNCS 4362, pages
  70--87, 2007.

\bibitem{FloPSW05}
P.~Flocchini, G.~Prencipe, N.~Santoro, and P.~Widmayer.
\newblock Gathering of robots with limited visibility.
\newblock {\em Theoretical Computer Science}, 337(1--3):147--168, 2005.

\bibitem{FloPSW08}
P.~Flocchini, G.~Prencipe, N.~Santoro, and P.~Widmayer.
\newblock Arbitrary pattern formation by asynchronous oblivious robots.
\newblock {\em Theoretical Computer Science}, 407(1--3):412--447, 2008.

\bibitem{ganguli09}
A.~Ganguli, J.~Cort\'es, and F.~Bullo.
\newblock Multirobot rendezvous with visibility sensors in nonconvex
  environments.
\newblock {\em IEEE Transactions on Robotics}, 25(2):340--352, 2009.

\bibitem{IIKO13}
T.~Izumi, T.~Izumi, S.~Kamei, and F.~Ooshita.
\newblock Feasibility of polynomial-time randomized gathering for oblivious
  mobile robots.
\newblock {\em IEEE Transactions on Parallel and Distributed Systems},
  24(4):716--723, 2013.

\bibitem{katreniak2011}
B.~Katreniak.
\newblock Convergence with limited visibility by asynchronous mobile robots.
\newblock In {\em Proceedings of the 18th International Colloquium on
  Structural Information and Communication Complexity (SIROCCO)}, pages
  125--137, 2011.

\bibitem{LinMA07b}
J.~Lin, A.S. Morse, and B.D.O. Anderson.
\newblock The multi-agent rendezvous problem---part 2: The asynchronous case.
\newblock {\em SIAM Journal on Control and Optimization}, 46(6):2120--2147,
  2007.

\bibitem{NearGSIROCCO}
L.~Pagli, G.~Prencipe, and G.~Viglietta.
\newblock Getting close without touching.
\newblock In {\em Proceedings of the 19th Colloquium on Structural Information
  and Communication Complexity (SIROCCO)}, LNCS 7355, pages 315--326, 2012.

\bibitem{pelegInvited}
D.~Peleg.
\newblock Distributed coordination algorithms for mobile robot swarms: New
  directions and challenges.
\newblock In {\em Proceedings of the 7th International Workshop on Distributed
  Computing (IWDC)}, LNCS 3741, pages 1--12, 2005.

\bibitem{SouDY09}
S.~Souissi, X.~D\'efago, and M.~Yamashita.
\newblock Using eventually consistent compasses to gather memory-less mobile
  robots with limited visibility.
\newblock {\em ACM Transactions on Autonomous and Adaptive Systems},
  4(1):1--27, 2009.

\bibitem{SuzY99}
I.~Suzuki and M.~Yamashita.
\newblock Distributed anonymous mobile robots: formation of geometric patterns.
\newblock {\em Siam Journal on Computing}, 28(4):1347--1363, 1999.

\bibitem{yamashita2010}
M.~Yamashita and I.~Suzuki.
\newblock Characterizing geometric patterns formable by oblivious anonymous
  mobile robots.
\newblock {\em Theoretical Computer Science}, 411(26--28):2433--2453, 2010.

\end{thebibliography}

\end{document}